%% file: Arxiv-main.tex
\documentclass[11pt]{article}

\usepackage[utf8]{inputenc}
\usepackage[margin=1in]{geometry}
\usepackage{graphicx} 

\graphicspath{{./figs/}}

\usepackage{amsthm,amsmath,amssymb,amsfonts}
\usepackage{tikz-cd}
\usepackage[colorlinks=true,linkcolor=black,anchorcolor=black,citecolor=blue,filecolor=black,menucolor=black,runcolor=black,urlcolor=black]{hyperref}
\usepackage{cleveref}
\usepackage{comment}
\usepackage{enumitem}

\newcommand {\mm}[1] {\ifmmode{#1}\else{\mbox{\(#1\)}}\fi}

\newcommand{\Rspace}        {\mm{\mathbb{R}}}

\newcommand{\har}           {\mm{\mathbb{H}}}

\newcommand{\chd}            {\mm{\text{CHD}}}

\newcommand{\eps}        {\varepsilon}

\newcommand{\etal}{{et al.}}

\newcommand{\hcd}{{harmonic chain barcode}}

\newcommand{\hcb}{{harmonic chain barcode}}
\newcommand{\hcbs}{{harmonic chain barcodes}}

\newcommand{\para}[1]        {\vspace{4mm}\noindent{\textbf{#1}}}

\newtheorem{lemma}{Lemma}

\newtheorem{theorem}{Theorem}
\newtheorem{definition}{Definition}
\newtheorem{proposition}{Proposition}

\newcommand{\update}[1]{{\color{black}{#1}}}

\crefname{equation}{Eq.}{Eqs.}
\crefname{figure}{Fig.}{Figs.}
\crefname{tabular}{Tab.}{Tabs.}
\crefname{section}{Sec.}{Secs.}
\crefname{appendix}{App.}{Apps.}

\theoremstyle{plain}

\title{Harmonic Chain Barcode and Stability}

\author{Salman Parsa\thanks{DePaul University, s.parsa@depaul.edu}
\and  Bei Wang\thanks{University of Utah, beiwang@sci.utah.edu.}}

\date{}

\begin{document}
%\linenumbers

\maketitle
%------------------------------
\begin{abstract}
\input{sec-abstract}

\end{abstract}
%------------------------------

\input{sec-introduction}

\input{sec-related-work}

\input{sec-background-long}
\input{sec-unstable-barcode-long}

\input{sec-stable-barcode-long}
\input{sec-conclusion}
%------------------------------
\bibliographystyle{abbrv}
\bibliography{refs-harmonic}
%------------------------------
\appendix  
\input{sec-correctness}

\end{document}

%% file: sec-abstract.tex
The persistence barcode is a topological descriptor of data that plays a fundamental role in topological data analysis. 
Given a filtration of the space of data, a persistence barcode tracks the evolution of its homological features.
In this paper, we introduce a new type of barcode, referred to as the canonical barcode of harmonic chains, or {\hcb} for short, which tracks the evolution of harmonic chains. 
As our main result, we show that the {\hcb} is stable and it captures both geometric and topological information of data. 
Moreover, given a filtration of a simplicial complex of size $n$ with $m$ time steps, we can compute its {\hcb} in $O(m^2n^{\omega} + mn^3)$ time, where $n^\omega$ is the matrix multiplication time.  
Consequently, a {\hcb} can be utilized in applications in which a persistence barcode is applicable, such as feature vectorization and machine learning. 
Our work provides strong evidence in a growing list of literature that geometric (not just topological) information can be recovered from a persistence filtration.

%% file: sec-introduction.tex
\section{Introduction}
\label{sec:introduction}

There are two primary tasks in topological data analysis (TDA)~\cite{dey2022computational,edelsbrunner2010computational,zomorodian2005topology}: reconstruction and inference. 
In a typical TDA pipeline, the data is given as a point cloud in $\Rspace^N$.  
A ``geometric shape'' $K$ in $\Rspace^N$ is reconstructed from the point cloud, usually as a simplicial complex, and $K$ is taken to represent the (unknown) space $X$ from which the data is sampled. 
$X$ is then studied using $K$ as a surrogate for properties that are invariant under invertible mappings (technically, homeomorphisms). 
The deduced ``topological shape'' is not specific to the complex $K$ or the space $X$, but is a feature of the homeomorphism type of $K$ or $X$. 
For example, a standard round circle has the same topological shape as any closed loop such as a knot in the Euclidean space. 
Although topological properties of $X$ alone are not sufficient to reconstruct $X$ exactly, they are among a few global features of $X$ that can be inferred from the data sampled from $X$. 

A persistence barcode~\cite{CarlssonZomorodianCollins2004,Ghrist2008} (or equivalently, a persistence diagram~\cite{Chazal, cohen2005stability, edelsbrunner2002topological}) captures the evolution of homological features in a filtration constructed from a simplicial complex $K$. 
It consists of a multi-set of intervals in the extended real line, where the start and end points of an interval (i.e.,~a bar) are the birth and death times of a homological feature in the filtration. 
Equivalently, a persistence diagram is a multi-set of points in the extended plane, where a point in the persistence diagram encodes the birth and death time of a homological feature. 
A main drawback of homological features is that they ignore the geometric information in $K$, even though they are defined on a set of simplicial chains with distinct geometric features. 

In this paper, we aim to recover geometric information from a filtration. 
We introduce a type of barcode, referred to as the \emph{canonical barcode of harmonic chains} or \emph{\hcb}, which tracks the evolution of harmonic chains in a filtration.  
A bar in a {\hcd} is associated with a single geometric feature, namely, a harmonic chain, and a {\hcb} tracks its birth and death in a filtration. 
To achieve this association, we need less choices (in terms of cycle representatives) than we need in ordinary persistence barcode. 
The main point of using harmonic chains is that a homology class contains a unique harmonic chain. 
Consequently, the homology class can be given the geometric shape of its harmonic chain without making a specific choice, giving rise to an interpretable, geometric feature of the data. 

A persistence barcode is stable~\cite{cohen2005stability}, and this stability is crucial for data science applications. The stability means that small changes in the data imply only small changes in the barcode. We show that our canonical {\hcb} is also stable.

Defining a canonical barcode of harmonic chains that is stable requires some care. To do so, we start by reviewing the computation of a persistence barcode.
Imagine that we are tracking the homology of a filtration of $K$ (with real coefficients $\Rspace$), constructed by adding simplices, one at a time.
At the filtration time $t$, we have maintained a set of homology \emph{basis}, that is, a maximal linearly independent set of homology cycles (of each dimension). 
Then a simplex $\sigma$ is added at time $t+1$ that destroys a homology class, thus reducing the dimension of a homology group. In other words, at time $t+1$, a linear combination $\tau$ of our basis elements becomes homologous to the boundary of the added simplex $\sigma$. 
The \emph{Elder Rule}~\cite{edelsbrunner2010computational} says that at time $t+1$, we destroy the youngest of the basis elements that appear in $\tau$. Other bars persist in the filtration.
Now, is it possible to canonically assign specific cycles (representatives) to a bar in the persistence barcode? 
Such an assignment is desirable, since then each bar would correspond to a unique geometric feature. 
However, as the Elder Rule is applied, we need to know what simplex is going to be inserted in the future to be able to choose a basis and a cycle that survives the insertion of simplices. 
Thus just by looking into the past, we cannot define stable and persistent geometric features corresponding to the persistence barcode.

There are two levels of choices to recover geometric features from a persistence barcode: first, we choose a basis for persistent homology in a consistent way across the filtration; and second, we choose a cycle inside a basis element of homology. 
Both levels of choices involve choosing among significantly distinct geometric features of data to represent the same bar in the barcode. 
A {\hcb} is defined using harmonic cycles, and immediately removes the second level of choices in a natural way, since there is always a unique harmonic cycle in a homology class.  

\para{Contributions.} Our contributions are as follows: 
\begin{itemize}[noitemsep,leftmargin=*]
\item We present canonical {\hcb}, a new type of barcode constructed from the same filtration used by a persistence barcode. Unlike a persistence barcode, a {\hcb} is defined by utilizing a global analysis of the filtration w.r.t. time, and it captures geometric (not just topological) information of data. 
\item We prove that the canonical {\hcd} is stable in the same setting as the persistence barcode. 
\item We introduce a \emph{harmonic interleaving distance} for the stability proof, which is of independent interest. 
\item We provide an algorithm for computing the {\hcd} that runs in $O(m^2n^{\omega} + mn^3) = O(m^2n^3)$ time, if the input complex is of size $n$, the filtration of the complex has $m$ time steps, and  $n^\omega$ is the matrix multiplication time. 
\end{itemize}
We expect a {\hcb} to be used in applications in which a persistence barcode is applicable, such as feature vectorization and machine learning. 
Our work also provides strong evidence in a growing list of literature (e.g.,~\cite{BubenikHullPatel2020,ChacholskiGiuntiJin2023}) that geometric (not just topological) information can be recovered from a persistence filtration.

%% file: sec-related-work.tex
\section{Related Work}
\label{sec:related-work}

Harmonic chains were first studied in the context of functions on graphs. 
They were identified as the kernel of the Laplacian operator on graphs~\cite{Kirchhoff1847}. 
The graph Laplacian and its kernel are important tools in studying graph properties, see~\cite{Merris1994,mohar1991laplacian} for surveys. 
Eckmann~\cite{eckmann1944harmonische} introduced the higher-order Laplacian for simplicial complexes, and proved the isomorphism of harmonic chains and homology. 
Guglielmi \etal~\cite{guglielmi2023quantifying} studied the stability of higher-order Laplacians. 
Horak and Jost~\cite{horak2013} defined a weighted Laplacian for simplicial complexes. 
Already their theoretical results on Laplacian~\cite{horak2013} anticipated the possibility of applications, as the harmonic chains are thought to contain important geometric information.
This has been validated by recent results that use curves of eigenvalues of Laplacians in a filtration in data analysis~\cite{chen2022, wang2020persistent}. 
The Laplacian was applied to improve the mapper algorithm~\cite{mike2019}, and for coarsening triangular meshes~\cite{keros2023spectral}. 
The persistent Laplacian~\cite{memoli2022persistent} and its stability~\cite{liu2023algebraic} is an active research area. 
Due to the close relation of harmonic chains and Laplacians, harmonic chains could find applications in areas that Laplacians have been used.

Computing reasonable representative cycles for persistent homology is also an active area of research. 
Here, usually an optimality criterion is imposed on cycles in a homology class to obtain a unique representative. 
For a single homology class, a number of works~\cite{chambers2022complexity, chen2011hardness, DeyHiraniKrishnamoorthy2010} consider different criteria for optimality of cycles. 
For persistent homology, Dey et al.~\cite{dey2020computing} studied the hardness of choosing optimal cycles for persistence bars. 
Volume-optimal cycles have been computed for persistent homology~\cite{Obayashi}. The harmonic chains have been used as representative of homology classes for studying the brain~\cite{lee2019}.
Furthermore, De Gregorio et al.~\cite{DeGregorioGuerraScaramuccia2021} used harmonic cycles in a persistent homology setting to compute the persistence barcode.  
Lieutier~\cite{Lieutier} studied the harmonic chains in persistent homology classes, called persistent harmonic forms.

The most relevant work is the one by Basu and Cox~\cite{basu2022harmonic}. 
They had a similar goal as we do, namely, to associate geometric information to each bar in a persistence barcode in order to obtain a more interpretable data feature. 
To that end, they introduced the notion of ``harmonic persistent barcode'', by associating a subspace of harmonic chains to each bar in the ordinary persistence barcode.
Then, they proved stability for their harmonic persistent barcode, by considering subspaces as points of a Grassmannian manifold and measuring distances in the Grassmannian. 
Their work was used to study multi-omics data~\cite{GurnariGuzman-SaenzUtro2023}. 
Different from the approach of Basu and Cox, we define a distinct barcode from the persistence barcode, which is stable in the same sense as the persistence barcode.
That is, the bottleneck distances between a pair of {\hcbs} is upper bounded by the harmonic interleaving distance.

%% file: sec-background-long.tex
\section{Background}
\label{sec:background}

In this section, we review standard notions of homology and cohomology with real coefficients and harmonic chains. 
We include some basic algebraic definitions to draw an understanding that helps with the rest of the paper.

\subsection{Homology and Cohomology}
Let $K$ be a simplicial complex. 
We give the standard orientation to the simplices of $K$. 
That is, we order the vertices of $K$ and assign to any simplex $\sigma \in K$ the ordering of its vertices induced from an ordering of the vertices. 
This is also the ordering used in all our examples. 
From now on, we think of a simplex $\sigma$ to be supplied with the standard orientation. 
We write simplices as an ordered set of vertices, for instance, if $s=v_0v_1\cdots v_{i}$, and $\sigma = \{v_0, v_1, \ldots, v_i \}$ then by convention $s=\sigma$ if the sign of the ordering $v_0v_1\cdots v_{i}$ agrees with the sign of the standard orientation, otherwise $s=-\sigma$.

For any integer $p \geq 0$, the $p$-dimensional \textit{chain group} of $K$ with coefficients in $\Rspace$, denoted $C_p(K)$, is an $\Rspace$-vector space generated by the set of $p$-dimensional (oriented) simplices of $K$. Let $K_p$ denote the set of $p$-simplices of $K$ and $n_p=|K_p|$. By fixing an ordering of the set $K_p$, we can identify any $p$-chain $c\in C_p(K)$ with an ordered $n_p$-tuple with real entries.
For each $p$, we fix an ordering for the $p$-simplices (once and for all), and identify $C_p(K)$ and $\Rspace^{n_p}$. 
The standard basis of $\Rspace^{n_p}$ corresponds (under the identification) to the basis of $C_p(K)$ given by the simplices with standard orientation.

The $p$-dimensional boundary matrix $\partial_p: C_p(K) \rightarrow C_{p-1}(K)$ is defined on a simplex basis element by the formula
\[ \partial (v_0v_1 \cdots v_p)= \sum_{j=0}^{p} {(-1)^q (v_0\cdots v_{q-1}v_{q+1}\cdots v_p) }.\]
In the right hand side above, in the $q$-th term $v_q$ is dropped. The formula guarantees the crucial property of the boundary homomorphism: $\partial_{p} \partial_{p+1} = 0$. This simply means that the boundary of a simplex has no boundary. The sequence $C_p(K)$ together with the maps $\partial_p$ define the \textit{simplicial chain complex} of $K$ with real coefficients, denoted $C_{\bullet}(K)$. The group of $p$-dimensional \textit{cycles}, denoted $Z_p(K)$, is the kernel of $\partial_p$. The group of $p$-dimensional boundaries, denoted $B_p(K)$, is the image of $\partial_{p+1}$. The $p$-dimensional homology group of $K$, denoted $H_p(K)$, is the quotient group $Z_p(K)/B_p(K)$. 
As a set, this quotient is formally defined as $\{z+B_p(K) \mid  z\in Z_p(K)\}$. The operations are inherited from the chain group. In words, the homology group is obtained from $Z_p(K)$ by setting any two cycles which differ by a boundary to be equal. All these groups are $\Rspace$-vector spaces.

Consider the space $\Rspace^{n_p}$. The cycle group $Z_p(K) \subset C_p(K) = \Rspace^{n_p}$ is a subspace, that is, a hyperplane passing through the origin. Similarly, the boundary group $B_p(K)$ is a subspace, and is included in $Z_p(K)$. The homology group is the set of parallels of $B_p(K)$ inside $Z_p(K)$. Each such parallel hyperplane differs from $B_p(K)$ by a translation given by some cycle $z$. These parallel hyperplanes partition $Z_p(K)$. The homology group is then isomorphic to the subspace perpendicular to $B_p(K)$ inside $Z_p(K)$. 
The dimension of the $p$-dimensional homology group is called the $p$-th $\Rspace$-Betti number, denoted $\beta_p(K)$. 

Simplicial \textit{cohomology} with coefficients in $\Rspace$ is usually  defined by the process of dualizing. This means that we replace an $i$-chain by a linear functional $C_p(K) \rightarrow \Rspace$, called a $p$-dimensional \textit{cochain}. 
The set of all such linear functions is a vector space isomorphic to $C_p(K)$, called the cochain group, denoted $C^p(K)$, which is a dual vector space of $C_p(K)$.  
For the purposes of defining harmonic chains, we must fix an isomorphism. 
We take the isomorphism that sends each standard basis element of $\Rspace^{n_p}$, corresponding to $\sigma$, to a functional which assigns 1 to $\sigma$ and 0 to other basis elements, denoted $\hat{\sigma} \in C^p(K)$. Any cochain $\gamma \in C^p(K)$ can be written as a linear combination of the $\hat{\sigma}$. Therefore, it is also a vector in $\Rspace^{n_p}$. The fixed isomorphism allows us to identify $C^p$ with $\Rspace^{n_p}$, and hence to $C_p$. Therefore, any vector in $\Rspace^{n_p}$ is at the same time a chain and a cochain.

The \textit{coboundary matrix} $\delta^p: C^p(K)\rightarrow C^{p+1}(K)$ is the transpose of $\partial_{p+1}$, $\delta^p = \partial_{p+1}^T$. It follows that $\delta_{p+1}\delta_p=0$, and we can form a \textit{cochain complex} $C^{\bullet}(K)$. The group of \textit{$p$-cocycles}, denoted $Z^p$ is the kernel of $\delta^p$. The group of \textit{$p$-coboundary} is the image of $\delta_{p-1}$. The $p$-dimensional \textit{cohomology group}, denoted $H^p(K)$ is defined as $H^p(K) = Z^p(K)/B^p(K)$. All of these groups are again vector subspaces of $C^p(K)$ and thus of $\Rspace^{n_p}$.
It is a standard fact that homology and cohomology groups with real coefficients are isomorphic.

\subsection{Harmonic Cycles} 
Recall that we identify chains with cochains. 
Therefore, we can talk about the coboundary of the cycles $Z_p(K)$. The \textit{harmonic $p$-cycles}, denoted $\har_p(K)$ is the group of those cycles which are also cocylces. Considered as subspaces of $\Rspace^{n_p}$, we have $\har_p(K) = Z^p(K) \cap Z_p(K)$.

\begin{lemma}[\cite{eckmann1944harmonische}]
    $\har_p(K)$ is isomorphic to $H_p(K)$. In other words, each homology class has a unique harmonic cycle in it.
\end{lemma}

Harmonic cycles enjoy certain geometric properties. As an example we mention the following. For a proof see~\cite[Proposition 3]{deSilvaVejedomo}.

\begin{proposition}
    Let $\alpha \in C^p(K)$ be a cochain. There is a unique solution $\bar{\alpha}$ to the least-squares minimization problem $$ \text{argmin}_{\bar{\alpha}} \{ || \bar{\alpha}||^2 \; | \; \exists \gamma \in C^{p-1}(K);  \alpha = \bar{\alpha}+\delta \gamma \}.$$ Moreover, $\bar{\alpha}$ is characterised by the relation $\partial {\bar{\alpha}}=0$.
\end{proposition}

In other words, the harmonic chain is the chain with the least squared-norm in a cohomology class.

\subsection{Persistent Homology}

Persistent homology tracks the changes in homology over time (or any other  parametrization).  
We start with a filtration $F$ of a simplicial complex $K$. A filtration assigns to each $r\in \Rspace$ a subcomplex $K_r \subset K$, in such a way that, if $r \leq s$, then $K_r \subseteq K_s$. 
Since $K$ is finite, there is a finite set of values $t_1, \ldots, t_m \in \Rspace$ where the subcomplex $F_{t_i}$ changes. Setting $K_i = K_{t_i}$, $K_0 = \emptyset$, $K_i \hookrightarrow K_{i+1}$ the inclusions, the filtration $F$ can be written as
\begin{equation}
\label{eq:fcomplex}
    \emptyset = K_0 \hookrightarrow K_1 \hookrightarrow \cdots \hookrightarrow K_{m-1} \hookrightarrow K_m=K.
\end{equation}

Applying homology functor to Eqn.~\eqref{eq:fcomplex}, we obtain a sequence of homology groups and connecting homomorphisms (linear maps), forming a \textit{persistence module}:
\begin{equation}\label{eq:fhomology}
 H_p(K_0) \xrightarrow{f_p^{t_0,t_1}} H_p(K_1) \xrightarrow{f_p^{t_1,t_2}}   \cdots \xrightarrow{f_{p}^{t_{m-2},t_{m-1}}} H_p(K_{m-1}) \xrightarrow{f_{p}^{t_{m-1},t_m}} H_p(K_m).   
\end{equation} 
In general, a persistence module $M$ at dimension $p$ is defined as assigning to each $r\in \Rspace$ a vector space $M_r$, and for each pair $s<t$, a homomorphism $f_p^{s,t}: M_s \xrightarrow{} M_t$. For all $s<t<u$, the homomorphism are required to satisfy linearity, that is, $f_p^{t,u} \circ f_p^{s,t}= f_p^{s,u}$. 

A cycle $z \in Z_p(K)$ appears for the first time at some index $b \in \{t_0,\ldots,t_m\}$, where it creates a new homology class in $K_b$ not previously existing. We say that the homology class is \textit{born} at time $b$. 
The homology class then lives for a while until it \textit{dies} entering $K_d$, when it lies in the kernel of the map induced on homology by the inclusion $K_b \subset K_d$. 
For $s\leq t$, let $f^{s,t}_p: H_p(K_s) \rightarrow H_p(K_t)$ denote the map induced on the $p$-dimensional homology by the inclusion $K_s \subset K_t$. 
The image of this homomorphisms, $f^{s,t}_p(H_p(K_s)) \subset H_p(K_t)$, is called the $p$-dimensional \textit{$(s,t)$-persistent homology group}, denoted $H_p^{s,t}$. The group $H_p^{s,t}$, which in our case is also a vector space, consists of classes which exist in $K_s$ and survive until $K_t$. The dimensions of these vector spaces are the \textit{persistent Betti numbers}, denoted $\beta^{s,t}_p$.

An \textit{interval module}, denoted $I=I(b,d)$, is a persistent module of the form
\[ 0 \rightarrow \cdots \rightarrow 0 \rightarrow \Rspace \rightarrow \Rspace \rightarrow 0 \rightarrow \cdots \rightarrow 0.\]
In the above, $\Rspace$ is generated by a homology class and the maps connecting the $\Rspace$s map generator to generator. 
The first $\Rspace$ appears at index $b$ and the last appears at index $d$. 
For any $b \leq r < d$, $I_r = \Rspace$; and for other $r$, $I_r=0$. 
Any persistence module can be decomposed into a collection of interval modules in a unique way \cite{LuoHenselman-Petrusek2023}. 
Consequently, a homology class that lives in time period $[b,d)$ can be written as a linear combination of intervals.
Each interval module is determined by the index of the first $\Rspace$, or the time it is born and the index of the last $\Rspace$ where it dies, which can be $\infty$.
The collection of $[b,d)$ for all interval modules is called the \textit{persistence barcode}. 
When plotted as points in an extended plane, the result is the equivalent \textit{persistence diagram}. 
The persistence barcode therefore contains a summary of the homological changes in the filtration. 

\para{Interleaving and Stability.}
Let $F$ and $G$ be two filtrations over the complexes $K$ and $K'$ respectively, and let the maps connecting the complexes of filtration be $f^{s,t}$ and $g^{s,t}$ respectively, for all $s,t \in \Rspace, s\leq t $. Let $C(F)$ and $C(G)$ be the corresponding filtrations of chain groups, and $H(F)$ and $H(G)$ the corresponding persistence modules. We denote the maps induced on chain groups and homology groups also by $f^{s,t}$ and $g^{s,t}$ respectively, for all $s,t \in \Rspace, s\leq t $.

\begin{definition}
\label{def:chain-interleaving}
Let $F$ and $G$ be two filtrations over the complexes $K$ and $K'$ respectively. 
An \textit{$\eps$-chain-interleaving between $F$ and $G$} (or an $\eps$-interleaving at the chain level) is given by two sets of homomorphisms $\{\phi_\alpha : C(K_\alpha) \xrightarrow{} C(K'_{\alpha+\eps}) \}$ and $\{\psi_\alpha : C(K'_{\alpha}) \xrightarrow{} C(K_{\alpha+\eps}) \}$, such that 
\begin{enumerate}
    \item $\{ \phi_\alpha \}$ and $\{\psi_\alpha \}$ commute with the maps of the filtration, that is, for all $\alpha,t \in \Rspace$, $g^{\alpha+\eps,\alpha+\eps + t} \phi_\alpha = \phi_{\alpha+t} f^{\alpha, \alpha+t}$ and $f^{\alpha+\eps,\alpha+\eps+t} \psi_\alpha = \psi_{\alpha+t} g^{\alpha, \alpha+t}$; 
    \item The following diagrams commute: 
    \begin{equation}
      \begin{tikzcd}
  C(K_\alpha)
    \ar[r,"f^{\alpha, \alpha+\eps}"]
    \ar[dr, "\phi_\alpha", very near start, outer sep = -2pt]
  & C(K_{\alpha+\eps})
    \ar[r, "f^{\alpha+\eps, \alpha+2\eps}"]
    \ar[dr, "\phi_{\alpha+\eps}", very near start, outer sep = -2pt]
  & C(K_{\alpha+2\eps})
  \\
  C(K'_\alpha)
    \ar[r, swap,"g^{\alpha, \alpha+\eps}"]
    \ar[ur, crossing over, "\psi_\alpha"', very near start, outer sep = -2pt]
  & C(K'_{\alpha+\eps})
    \ar[r, swap,"g^{\alpha+\eps, \alpha+2\eps}"]
    \ar[ur, crossing over, "\psi_{\alpha+\eps}"', very near start, outer sep = -2pt]
  & C(K'_{\alpha+2\eps}).
  \end{tikzcd}
    \end{equation}
\end{enumerate}

The \emph{chain interleaving distance} is defined to be 
\begin{equation}
d_{CI}(F, G) := \inf \{\eps\geq 0 \mid \text{there exists an $\eps$-chain interleaving between } F\; \text{and}\; G\}.
\end{equation}
\end{definition}

The standard notion of $\eps$-interleaving~\cite{ChazalCohenSteinerGlisse2009}, denoted $d_I(F,G)$, is defined analogously to our definition above, however, it is defined on the persistence modules $H(F)$ and $H(G)$ on the homology level, rather than the filtration of chain groups.  For our purposes, we require the existence of a chain interleaving. In the application to sublevel set filtrations, which is the main setting in which stability is proved~\cite{cohen2005stability}, the chain interleaving exists, therefore, this strengthening of the interleaving does not hurt the sublevel set stability arguments.

The stability of persistence barcode (or diagram) is a crucial property for applications. It expresses the fact that small changes in data lead to small changes in persistence barcodes. The interleaving distance provides the measure of change in data, and we measure the distance between barcodes using the bottleneck distance. Let $D=Dgm(F)$ and $D'=Dgm(G)$ denote the persistence barcodes (or diagrams) of filtrations $F$ and $G$. 
Recall that the diagram is a multi-set of points and contains all the diagonal points.
The \textit{bottleneck distance} is defined as 
\[
d_B(D, D') = \text{inf}_\gamma \text{sup}_{p \in D} || p - \gamma(p) ||_\infty,
\]
where $\gamma$ ranges over all bijections between $D$ and $D'$ and $||\cdot ||_\infty$ is the largest absolute value of the difference of coordinates, or considered as bars, the largest distortion of an endpoint when matched to the bar of $D'$.   
We refer to~\cite{ChazalCohenSteinerGlisse2009} for the proof of the following \Cref{theorem:persistence-stability}.

\begin{theorem}
\label{theorem:persistence-stability}
    Let $F$ and $G$ be filtrations defined over (finite) complexes $K$ and $K'$, respectively. Then $$ d_B( Dgm(F), Dgm(G)) \leq d_I(F,G).$$
\end{theorem}

Please see~\cite{cohen2005stability, ChazalCohenSteinerGlisse2009} for more on stability. We also mention that in the above theorem we can replace $d_{CI}$ in place of $d_I$, since existence of an $\eps$-chain-interleaving implies existence of an $\eps$-interleaving on the homology level. The resulting theorem would be weaker.

\para{Conventions.}
Since we work with chains and apply boundary and coboundary operators at different times, we write $\delta_t(c)$ and $\partial_t(c)$ to mean the coboundary and boundary of the chain $c$ at time $t$, respectively, where we consider $c$ to be a chain in the final complex $K$ present at time $t$. 
In this paper, we usually omit the subscript of a homology group; we always fix a dimension and do not write it if there is no danger of ambiguity.
Moreover, we use the words persistence barcode and persistence diagram interchangeably. 
A persistence diagram contains all the diagonal points $[0,0)$, and a persistence barcode contains infinite number of intervals of length $0$. 
We use $F$ and $G$ to denote two filtrations over the complexes $K$ and $K'$ respectively, and we denote the maps connecting the complexes of filtration by $f^{s,t}$ and $g^{s,t}$ respectively, for all $s,t \in \Rspace, s\leq t $. We denote the maps induced on chain groups and homology groups also by $f^{s,t}$ and $g^{s,t}$ respectively, for all $s,t \in \Rspace, s\leq t$.

%% file: sec-unstable-barcode-long.tex
\section{A First Attempt: Unstable Harmonic Chain Barcode}
\label{sec:unstable-barcode}

In this section, we discuss our first attempt at constructing a {\hcb} from a persistence barcode. 
Our approach is a natural one, however, we show that it does not lead to a stable barcode of harmonic chains. 
Nevertheless, our first attempt sheds light on how we might search for one that is stable, as discussed in~\cref{sec:stable-barcode}.

To avoid confusions, we use the following notations. 
Ordinary persistent homology gives rise to \emph{persistence barcodes}. 
An interval belonging to a persistence barcode is referred to as a \emph{persistence bar}. 
In our setting, we introduce \emph{\hcbs}. 
An interval belonging to a {\hcb} is called a \emph{harmonic bar}. 

\subsection{Constructing a Harmonic Chain Barcode from a Persistence Barcode}
\label{sec:first-attempt}

Starting from a persistence barcode, we assume there exists a choice of a homology class for each persistence bar, such that at each time $t$, the classes of existing bars at time $t$ form a basis of the homology group of $K_t$. There are many such choices, one such choice is represented by a function $\phi$ that maps persistence bars to homology classes as follows.
Let $B = [b, d)$ be a single interval (persistence bar) in the persistence barcode. 
$\phi$ maps each persistence bar $B$ to a chosen homology class, denoted by $\phi(B)$. The start of the interval $B$ is the birth of $\phi(B)$. 
Let $z_b$ be the harmonic cycle in the homology class $\phi(B)$. 
Starting from $b$, we follow through time until a time $s_1$ such that $\delta_{s_1}(z_b)\neq 0$. This implies that $\delta_{s_1}(z_b) = \sum \ell_j \sigma_j$, for some $\ell_j \in \Rspace$, where $\sigma_j$ are the simplices inserted at time $s_1$. 
At this point, the cycle $z_b$ dies entering $s_1$, and there is a bar $[b,s_1)$ in the {\hcb}. If $s_1 \neq d$, there has to be another cycle $z_1$ such that $[z_1] = [z_b]$ at time $s_1$, and $z_1$ is a cocycle in $K_{s_1}$. 
We therefore generate a new harmonic bar with a start time of $s_1$, and associate the cycle $z_1$ to it, and so on. After processing all (finite filtration) times, we obtain a {\hcb} subordinate to $\phi$. 

The above algorithm decomposes the persistence bars into a type of harmonic bars. The decomposition is guided by the choice of basis elements for each persistence bar by $\phi$. Each harmonic bar represents a unique harmonic cycle, hence, contains geometric information. 

\subsection{Example}

\cref{fig:first-attempt-example} presents an example where we compute the {\hcb} starting with the cycles for the persistence bars obtained with the standard matrix reduction algorithm for persistence homology. 
The matrix reduction algorithm to obtain persistence barcodes runs in the worst case $O(n^3)$, where $n$ is the number of simplices~\cite{Morozov2005}. 
We use the time stamps as the names of vertices and edges (from $1$ to $15$), and $t_i$ as the names of triangles ($t_1$ to $t_4$). 
We give simplices the standard orientation determined by the ordering of the vertices given by the insertion time. 
The procedure we describe in~\cref{sec:first-attempt} gives an algorithm that runs in $O(n^4)$ to extract a {\hcb} from the persistence barcode, since it makes a pass of the filtration in $O(n)$ time with a persistence barcode computed in $O(n^3)$ time. 

\begin{figure}[!ht]
    \centering
    \includegraphics[scale = 0.75]{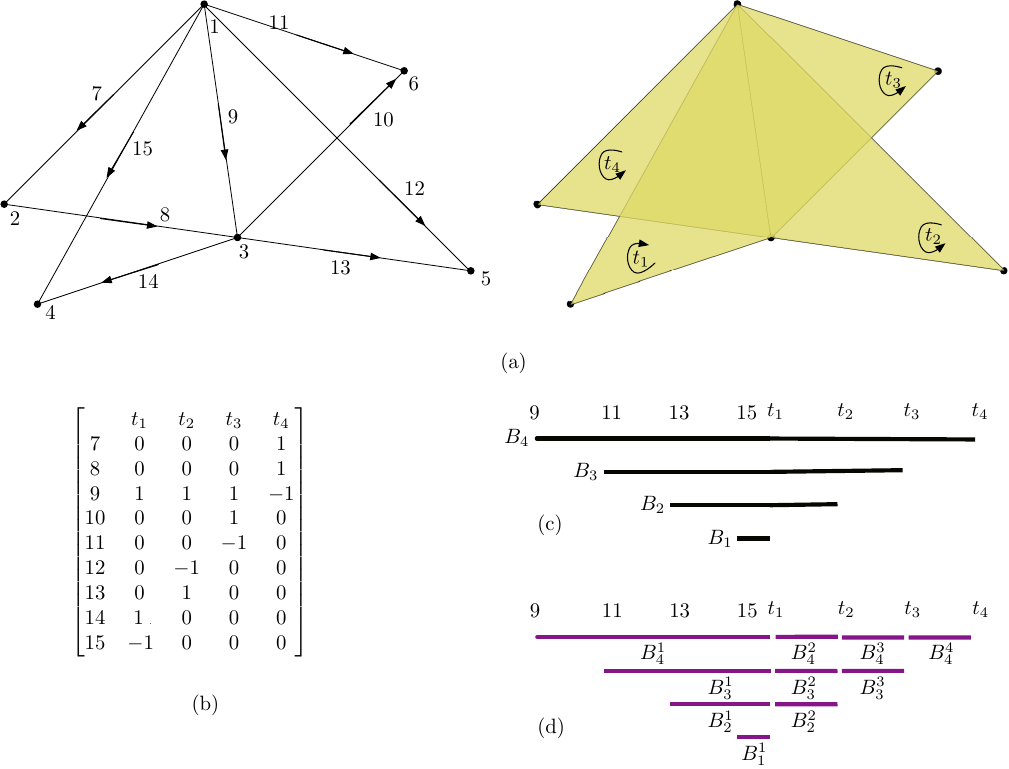} 
    \caption{(a) A filtration of a simplicial complex, the time that a simplex is inserted is written next to the simplex. For the four triangles we have times $t_1 < t_2 < t_3 < t_4$. We also use these time stamps as  names of the simplices. (b) A part of the boundary matrix of the complex which is relevant to 1-dimensional homology. The boundary matrix is already reduced, therefore, the generators of persistent homology computed by the standard matrix reduction algorithm are given by columns of the matrix. (c) The 1-dimensional ordinary persistence barcode. (d) The 1-dimensional {\hcb} based on the cycles computed by the matrix reduction algorithm.}
    \label{fig:first-attempt-example}
\end{figure}

%------------
% Check 
%\delta(z_4^3) = 0
%= \delta(z_4) + 1/4 \delta(\partial t_1 + \partial t_2)
%= \delta(7+8-9) + 1/4 \delta(9+14-15 + 9+13-12)
%= 0 + 0 - (t_1 + t_2) + 1/4 (t_1 + t_2 + t_1 + t_1 + t_1 + t_2 + t_2 + t_2)
% = 0 
%------------

Following the filtration given in~\cref{fig:first-attempt-example}(a), the 1-dimensional persistence barcode tracks the births and deaths of 1-cycles. After all vertices and edges have been inserted at time $15$, the triangles $t_1$, $t_2$, $t_3$, and $t_4$ destroy the 1-cycles created at time $15$, $13$, $11$, and $9$, respectively. This gives rise to a persistence barcode (in black) that consists of intervals \update{$B_1=[15,t_1),B_2=[13,t_2),B_3=[11,t_3),B_4=[9,t_4)$}, respectively, shown in~\cref{fig:first-attempt-example}(c).  
\update{After matrix reduction in~\cref{fig:first-attempt-example}(b), 
$B_1$ is represented by the 1-cycle $z_1=9+14-15$, 
$B_2$ is represented by $z_2=13-12+9$,
$B_3$ is represented by $z_3 = 9+10-11$, 
and $B_4$ is represented $z_4 = 7+8-9$.}

To obtain harmonic bars (in purple) we sweep from left to write. 
Right before time $t_1$, $z_i$ also serve as the harmonic representatives.  
At time $t_1$, all of the basis cycles $z_i$ get a coboundary from the triangle $t_1$. 
The homology class of \update{$z_1$} becomes trivial as it is destroyed by the triangle $t_1$, whereas other three homology classes remain non-trivial. 
As the first harmonic bar we observe, $B_1^1 = [15, t_1)$ is identical to the persistence bar $B_1$ and it is represented by a harmonic 1-cycle \update{$z_1^1=z_1$}. 
Now, if we take the cycle representative \update{$z_4$}, the homology class it represents remains non-trivial at time $t_1$, whereas the harmonic 1-cycle $z_4^1=z_4$ has been destroyed creating a harmonic bar $B_4^1=[9,t_1)$.   
We then need to find a new harmonic cycle in its homology class, by finding $x$ such that \update{$\delta(z_4 + \partial x) =0$} at time $t_1$. 
This implies $\delta(z_4) = - \delta \partial x$. 
Since only triangle $t_1$ exists at this time, $\delta(z_4) = \delta(7+8-9) = \delta(7)+\delta(8)- \delta(9) = 0+0-t_1=-t_1.$ 
We also have $\delta \partial t_1 = \delta(15-14+9) = \delta(15)-\delta(14)+\delta(9) = t_1 + t_1 + t_1 = 3t_1$. 
Therefore we can set $x = \frac{1}{3} (t_1)$. 
Then the new harmonic chain is $z^2_4 = z_4 + \frac{1}{3} \partial (t_1)$. 
Similarly, persistence bars $B_2$ and $B_3$ are split at time $t_1$, giving rise to harmonic bars with representatives $z^2_2 = z_2 - \frac{1}{3} \partial (t_1)$ and $z^2_3 = z_3 - \frac{1}{3} \partial (t_1)$. 

At time $t_2$, \update{$z_2^2$} dies as it is destroyed by the triangle $t_2$. 
\update{$z_4^2$} gets a coboundary $-t_2$ and $\delta t_1$ gets a coboundary $t_2$, hence $\delta (z_4^2) = -\frac{2}{3} t_2$. 
A new harmonic cycle is $z_4^3 = z_4 + \frac{1}{4}(\partial t_1 + \partial t_2)$. %\Bei{why?} 
Similarly, $z_3^2 = z_2 - \frac{1}{4}(\partial t_1 + \partial t_2)$. 
At time $t_3$, the cycle $z_3^3$ and its harmonic homologous cycles die. $z_4^3$ also becomes non-harmonic and will be replaced by another harmonic cycle that will die at $t_4$. 

Although in this example, each harmonic chain becomes non-harmonic after the insertion of a simplex, this is not a general phenomenon. 
The harmonic bars can be longer than one time step.

\subsection{Instability}

In a {\hcb} subordinate to some choice of basis element $\phi$ constructed above, the harmonic bars can appear and disappear without much restrictions.
Therefore, it is expected that such {\hcbs}
are not stable. 
Here we demonstrate this instability by presenting an example that shows when the basis elements are given by the standard reduction algorithm of persistence homology, the resulting {\hcb} is not stable. 

\begin{figure}[!ht]
    \centering
    \includegraphics{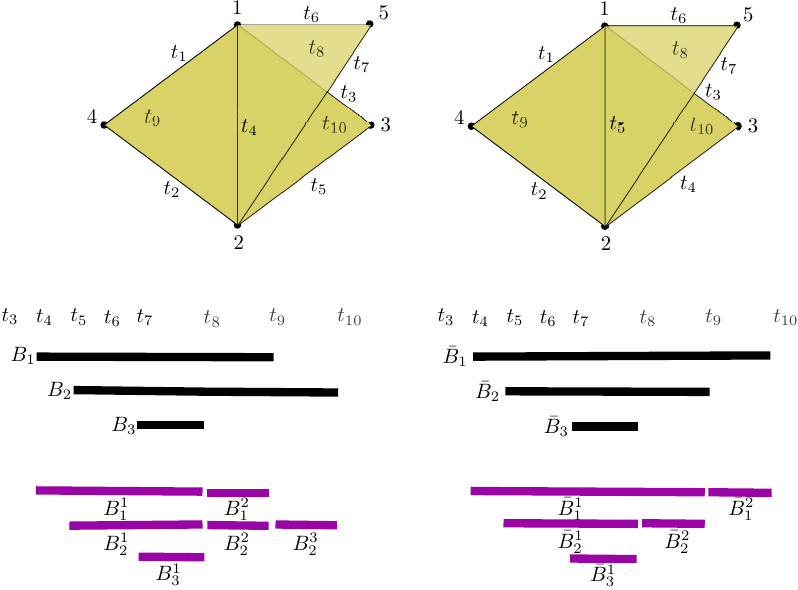}
    \caption{Changing the order of two simplices might effect an arbitrary large change in the harmonic barcode. The filtration differ by exchange of two edges. The ordinary persistence barcode is in black, whereas the {\hcb} is in purple (constructed using the basis given by the standard matrix reduction algorithm).}
    \label{fig:instability}
\end{figure}

\cref{fig:instability} shows two filtrations whose only difference is the exchange of times when edges $12$ and $23$ are inserted. 
In this example, vertices are labeled 1 to 5, and higher-dimensional simplices (edges and triangles) are labeled by $t_i$ (for some $i$). 
That is, edges appear at times $t_1$ to $t_7$, and triangles appear at times $t_8,t_9$ and $t_{10}$, respectively. 
The persistence barcodes are shown in black. 
The unstable {\hcbs} constructed from our first attempt using the basis given by the standard matrix reduction algorithm are shown in purple. 

In the left filtration, the matrix reduction algorithm for persistence homology gives the following cycle representatives,   
$z_1 = -t_1+t_2+t_4$, 
$z_2 = -t_3+t_4+t_5$,
and $z_3=t_4-t_6+t_7$, 
for persistence bars $B_1$, $B_2$, and $B_3$, respectively. 
In contrast, the harmonic bars are shown in purple. 
Before time $t_8$, $z_1^1=z_1, z_2^1=z_2$ and $z_3^1=z_3$ serve as harmonic cycle representatives. 
At time $t_8$, $z_3$ dies as the triangle $t_8$ enters the filtration. 
And the addition of triangle $t_8$ causes $z_1$ and $z_2$ to have coboundary. 
A new cycle homologous to $z_1$ and harmonic at $t_8$ can be found by finding a solution to $\delta (z_1 + \partial x) =0$. 
This means $\delta \partial x = -\delta z_1 = -t_8$. 
Since $\delta \partial t_8 = 3t_8$, we can take $x=-\frac{1}{3} t_8$. 
The new harmonic cycle is then $z_1^2 = z_1 - \frac{1}{3}(\partial t_8)$. At time $t_9$, $z_1$ dies and clearly $z_1^2$ also dies. At time $t_8$, $z_2$ dies and similarly we get a new cycle $z_2^2 = z_2 - \frac{1}{3}(\partial t_8)$. This cycle dies at $t_9$ and a new cycle replacing it dies at $t_{10}$. 

In the right filtration, we again have cycles computed by the matrix reduction algorithm, from top to bottom, $\bar{z}_1 = t_1 - t_2 - t_3 +t_4$, $\bar{z}_2 = -t_1+t_2+t_5$ and $\bar{z}_3 = t_5 - t_6 + t_7$ serve as representatives for $\bar{B}_1, \bar{B}_1$, and $\bar{B}_3$, respectively. The harmonic bars are obtained using these basis elements in a similar fashion. At time $t_8$, the cycle $\bar{z}_1$ remains harmonic since it is not incident to $t_8$. 
As before, $\bar{z}_2$ becomes non-harmonic at $t_8$ and is replaced by a harmonic cycle that dies at $t_9$. 
The cycle $\bar{z}_3$ dies also at $t_8$. 

A matching provided by the stability of the persistence diagram needs to match $\bar{B}_2$ with $B_1$, and $\bar{B}_1$ with $B_2$. 
Now, even if $\bar{z}_1^1=\bar{z}_1$ dies immediate at $t_9$, its bar will have a surplus of length of at least $t_9 - t_8$ and this can be made arbitrarily large compared to $t_5-t_4$. 
Hence it cannot be matched to any bar on the left {\hcb} without incurring an arbitrary high cost to the matching. 
In summary, our first attempt at constructing {\hcbs} leads to unstable ones; more constraints need to be imposed to construct stable {\hcbs}, which is discussed in~\cref{sec:stable-barcode}.

%% file: sec-stable-barcode-long.tex
\section{Canonical Harmonic Chain Barcode and Stability}
\label{sec:stable-barcode}

Fix a homological dimension $p$ and consider a filtration $F$ and a homology class $h \in H_p(K_b)$ born at time $b$ which dies at time $d$ when entering $K_d$. 
The rank (dimension) of $p$-th homology group increases at time $b$. The class $h$ contains a unique harmonic cycle $z$. We say that the harmonic cycle $z$ is \emph{born} at time $b$. When a persistent homology class dies (i.e., becomes a boundary), then clearly the corresponding harmonic cycle also \emph{dies}. Contrary to the ordinary persistent homology, this is not the only situation where a harmonic cycle dies. 
In brief, a harmonic cycle dies whenever it bounds a chain or gets a non-zero coboundary. 
A crucial property, which is not hard to verify, is that when a cycle is no longer a cocycle, it could never become a cocycle in the future. Moreover, when a harmonic cycle dies, it is because a relation is introduced in homology (i.e. a homology class dies); the classes then change their harmonic representatives.  

Recall that the space of harmonic $p$-cycles at time $t$, which we denote $\har_p(K_t)$, is isomorphic with the $p$-th homology $H_p(K_t)$. A cycle might be a harmonic cycle when it is first created, and then it might fail to be harmonic at a later time: for example, it might get a non-zero coboundary,  or become a boundary. Alternatively, it may remain harmonic up to infinity. 
When working with cycles, we usually fix a dimension $p$ and  suppress the dimension in the notation $H:=H_p, Z:=Z_p$, etc. 
Moreover, for simplicity of exposition, we treat $\infty$ as a fixed large number, so that all our interval lengths can be compared to each other.

\begin{definition}
\label{def:span}
For a cycle $z \in Z(K)$, its \emph{harmonic span}, denoted $span(z)$, is an interval of the form $[s,t)$ in the filtration during which $z$ is a (non-trivial) harmonic cycle. 
If a cycle is non-harmonic at birth, its span is empty and has length $0$. If a cycle is harmonic in $K$, its span has length $\infty$. The length of a span is denoted by $|span(z)|$.   
\end{definition}

\Cref{def:span} assigns an interval to each cycle of the final complex $K$. This would be an assignment of an interval to each vector in a $\Rspace$-valued vector space. When talking about independence of cycles, we always consider them as vectors in the cycle space $Z(K)$. 

\begin{definition}
\label{def:independent-cycles}
Let $\big(z_i\big):=z_1,z_2,\ldots \in Z(K)$ be a (finite) sequence of cycles. 
We say $\big(z_i\big)$ is a \emph{sequence of independent cycles} if, for each $i>1$, $z_i$ does not belong to the subspace generated by $z_1, \ldots, z_{i-1}$ in $Z(K)$. 
If $\big(z_i\big)$ is a sequence of independent cycles, the sequence 
$\big(span(z_i)\big): =span(z_1), span(z_2), \ldots$ is called a sequence of \emph{independent spans}, and the sequence 
$\big(|span(z_i)|\big):= |span(z_1)|, |span(z_2)|, \ldots$ is called a sequence of \emph{independent span-lengths}.  
The set of spans forming a sequence of independent spans is called a \emph{set of independent spans}.
\end{definition}

The number of cycles in a sequence of independent cycles equals the dimension of $Z(K).$
A sequence of independent span-lengths is a sequence of real numbers. We can compare any two sequences using the lexicographic ordering, starting from the first number in the sequence. 
That is, if $S_1 = \big(|span(z_i)|\big)$ and $S_2 =  \big(|span(y_i)|\big)$ are two sequences of independent span-lengths, we say $S_1 > S_2$ if and only if there is an index $k$ such that for all $j<k$, $|span(z_i)|=|span(y_i)|$, and $|span(z_k)| > |span(y_k)|$.  

\begin{theorem}
\label{thm:uniqueness}
    Let $\big(\lambda_i\big) = \lambda_1, \lambda_2, \ldots $ be the sequence of independent span lengths which is lexicographical maximal. Then there is a unique set of independent spans $\{J_1, J_2, \ldots\}$ whose sequence of lengths equals $\big(\lambda_i\big)$.  
\end{theorem}

\begin{proof}
We prove a stronger claim below by induction on the dimension of the subspaces of Z(K):
\begin{itemize}[noitemsep,leftmargin=*]
\item 
For each subspace $A \subset Z(K)$, there is a unique set of independent spans of cycles of $A$ realizing a lexicographical maximal (lex.-maximal) sequence of span lengths of cycles of $A$.
\end{itemize}
 
The base case for the induction is to prove the theorem for any subspace $A \subset Z(K)$ of dimension 1. $A$ is generated by any of its non-zero elements $z$. If $\delta(z)\neq 0$, then $\delta(cz)\neq 0$ for all $c\neq 0$. Moreover, if $\partial(d) = z$, $\partial (cd) = cz$. Therefore, all of $A$ either die at the same time or survive until the end. In each case, the spans of all cycles are equal, and there is exactly one span.

Assume the statement is true for all subspaces of dimension at most $m$, and we consider a subspace $A \subset Z(K)$ of dimension $m+1$. For the sake of contradiction, let $I_1, I_2, \ldots$ and $J_1, J_2, \ldots$ be two different sequences of independent spans realizing $\lambda_1, \lambda_2, \ldots$ which is the lex.-maximal sequence of independent span lengths of cycles in $A$. 

Let $k$ be the first index such that $I_k \neq J_k$. Let $\lambda_{l_1}, \ldots, \lambda_{l_s}$ be all the $\lambda_i$ such that $\lambda_{l_j} = \lambda_k$ (for $j=1,\ldots,s$). 
If the two sets of spans $\{I_{l_i}, i =1,\ldots,s \}$ and $\{ J_{l_i}, i=1, \ldots, s \}$ are equal, then, as sets, the two sequences would not differ at spans of length $\lambda_k$. We can therefore assume that these two sets $\{I_{l_i}\}$ and $\{J_{l_i}\}$ are distinct. For each $k$, let $I_k$ be the span of a cycle $x_k$ and $J_k$ be the span of a cycle $y_k$. $\big(x_k\big)$ and $\big(y_k\big)$ are sequences of independent cycles by definition.

We divide the argument into three cases:
\begin{itemize}[noitemsep,leftmargin=*]
\item Case (1) $l_s$ is less than $\dim(A)$. This means that there are independent intervals with length smaller than $\lambda_k$. The subspace $S_x \subset A$ generated by $x_1, \ldots, x_{l_s}$ and $S_y \subset A$ generated $y_1, \ldots, y_{l_s}$ are of dimension at most $m$. Since the sets $\{I_{l_i}, i =1,\ldots,s \}$ and $\{ J_{l_i}, i=1, \ldots, s \}$ are distinct, and these are the only ones of length $\lambda_k$ in the two sequences, the sets $\{ x_{1}, \ldots, x_{l_s} \}$ and $\{y_1, \ldots, y_{l_s}\}$ are distinct. The induction hypothesis implies that $S_x$ and $S_y$ are distinct subspaces of equal dimensions. It follows that at least one of $y_1, y_2, \ldots, y_{l_s}$ is not in $S_x$. Let $y$ be such a cycle. Adding $y$ to sequence $x_1, \ldots, x_{l_s}$ certainly creates a larger lexicographic sequence of lengths of independent cycles. This is a contradiction.   
\item Case (2) $l_s$ equals $\dim(A)$ and $k>1$. In this case, let subspace $S_x \subset A$ be generated by $x_1, \ldots, x_{l_1-1}$ and $S_y \subset A$ be generated $y_1, \ldots, y_{l_1-1}$. If $S_x$ and $S_y$ are distinct subspaces, then at least one of $y_1, \ldots, y_{l_1-1}$, say $y$, would not be in $S_x$. Adding length of span of $y$ to the maximal sequence of independent span lengths will make it larger. It follows that $S_x = S_y$. Therefore, all the spans of length less than $\lambda_k$ in $(I_i)$ and $(J_i)$ are equal as sets. Let $T = \{ span(z) | z \in A ,  z \notin S_x \}$. By our assumption, the sets $\{ I_{l_1}, \ldots, I_{l_s} \}$ and $\{ J_{l_1}, \ldots, J_{l_s} \}$ are two distinct subsets of $T$. If they generate the same subspace, the subspace would be of dimension at most $m$, and they have to be equal, by the induction hypothesis, since they realize the same maximal span-length sequence, namely $l_s$ copies of $\lambda_k$. Hence, they generate distinct subspaces which again leads to a contradiction similar to Case (1).
\item Case (3) $l_s$ equals $\dim(A)$ and $k=1$. The spans of length $\lambda_k$ create all of $A$ and all the $\lambda_i$'s are equal. The two sequences $\big(x_k\big)$ and $\big(y_k\big)$ give two sets of basis elements for $A$, all having equal span-lengths. Let $x$ be a cycle with smallest birth time whose interval is not among $\big(J_i\big)$. Let $t$ be the time $x$ is generated. If no $y_k$ is generated at $t$, then $x$ is independent of all $y_k$ which is a contradiction. Therefore some $y_k$, say $y$, is generated at time $t$. Since $span(x)$ has the same length as $span(y)$, it follows that $span(x) = span(y)$, and we reach a contradiction with the choice of $x$. 
\end{itemize}
\end{proof}
\noindent \Cref{thm:uniqueness} allows us to define the \emph{canonical barcode of harmonic chains}  unambiguously as follows.

\begin{definition}
\label{def:chcb}
Given a filtration $F$, the set of spans realizing the lexicographical maximal sequence of independent span lengths is called the \emph{canonical barcode of harmonic chains} (equivalently, the \emph{canonical diagram of harmonic chains}), denoted as $\chd(F)$. 
\end{definition} 
Given any filtration $F$, \Cref{def:chcb} associates a unique barcode to $F$. And this barcode is, in general, distinct from a persistence barcode.
For now on, canonical harmonic chain barcode (or simply,  {\hcb}) means canonical barcode of harmonic chains. 

\subsection{Example 1}
\cref{fig:c-harmonic-example} depicts a $1$-dimensional canonical barcode of harmonic chains, from the example in~\cref{fig:first-attempt-example}. 
At time $9$, the first cycle $r_1 = 7+8-9$ is created. 
This event creates a harmonic bar that ends at $t_1$ when $r_1$ gets a coboundary. 
At time $11$, a second cycle $r_2$ is created, leading to an increase in the dimension of $Z(K_{11})$. 
To compute a cycle that survives the longest, we start from $t=+\infty$ and compute the largest $t$ such that $Z(K_{11}) \cap \har(K_t) = Z(K_{11}) \cap \ker (\delta_t) \neq \{ 0 \}$. Take the cycle $r_2 = 7+8+10-11$. $r_2$ is harmonic, and it survives up until $t_3$. 
One can check that if a cycle has non-zero coefficient for $9$, it has to die earlier; therefore, $r_2$ represents the longest bar generated at $11$, namely, $[11,t_3)$. 
At time $13$, similar to time $11$, we obtain a cycle $r_3$ which represents a bar $[13,t_2)$. 
Now, at time $15$, we have all possible edges so there has to be a harmonic chain that survives up to $t_4$. We take $r_4 = 3r_1 - \partial( t_1+t_2+t_3)$. One checks easily that $\delta_t \partial (r_j) = \delta_t(9)$, for $j=2,3,4$. Therefore, $\delta_t(r_4) = 3\delta_t(9)-3\delta_t(9)=0$ for $ 15 \leq t < t_4$. 

\begin{figure}
    \centering
    \includegraphics[scale=.85]{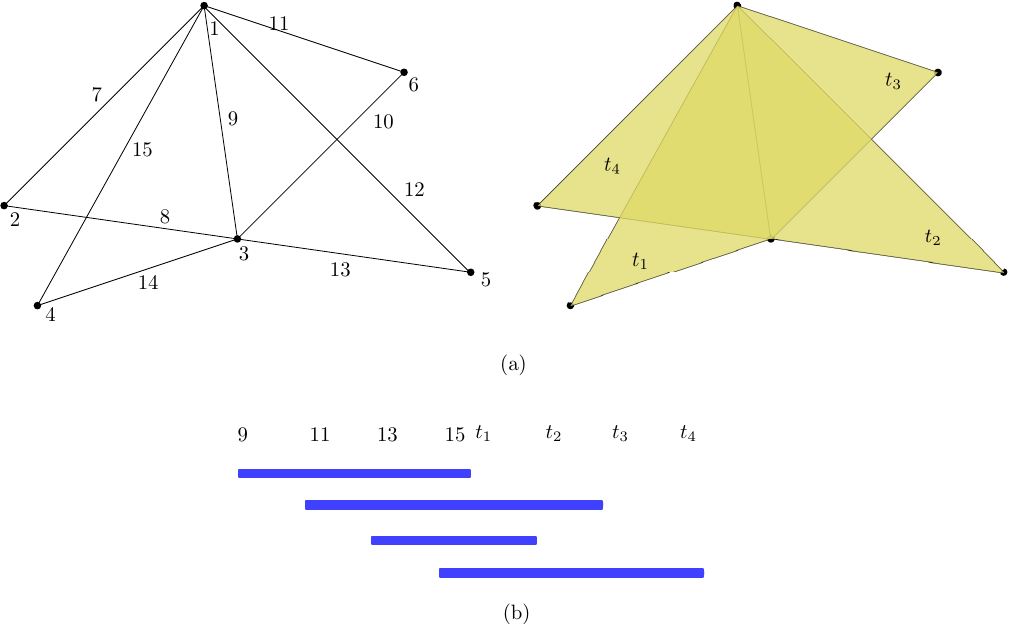}
    \caption{The canonical {\hcb} (in blue) for the filtration  in~\cref{fig:first-attempt-example}.}
    \label{fig:c-harmonic-example}
\end{figure}

\subsection{Example 2}
We provide a second example which sheds more light on the nature of the canonical {\hcb}. In this example, we have two filtrations $F_1$ and $F_2$ of the same underlying complex $K$. We are only interested in the 1-dimensional homology. In $F_1$, two independent 1-dimensional  homology classes are born early in the filtration, then cycles of low persistence are created and immediately killed. At the end, the two independent cycles created at the start are also killed. 

\begin{figure}[!ht]
    \centering
    \includegraphics[scale=0.8]{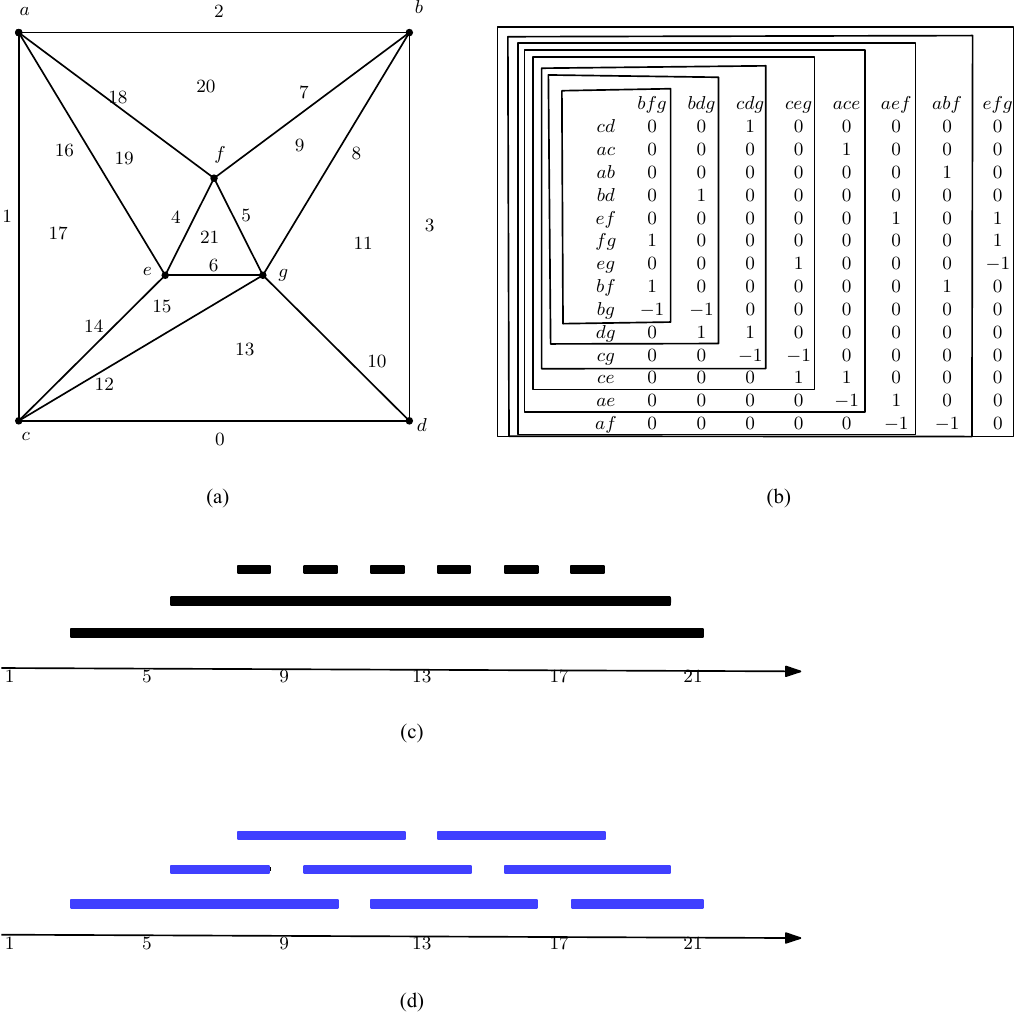}
    \caption{The ordinary persistence barcode and the canonical barcode of harmonic chains for the filtration $F_1$. (a) The insertion times of of edges and triangles.  (b) The 2-dimensional boundary matrix, the matrices of different times steps are marked. (c) The ordinary persistence barcode (in black). (d) The canonical barcode of harmonic chains (in blue).}
    \label{fig:example-2a}
\end{figure}

The first filtration $F_1$ is shown in~\cref{fig:example-2a}. 
Its ordinary persistence barcode is shown in (c), and its canonical {\hcb} is shown in (d). 
We see that, in the ordinary persistence barcode in (c), the two 1-dimensional homology classes that are created at the beginning give rise to two long bars, whereas other cycles create short bars. 
Contrary to the ordinary persistence barcode, we see that all the harmonic bars in the canonical {\hcb} (d) are rather short.

In~\cref{fig:example-2b}, we consider a second filtration $F_2$. 
In this filtration, first all the edges are inserted and then all the triangles. 
The canonical {\hcb} is justified since at the time when the last edge is inserted, already all the harmonic chains are present, and at the insertion of a triangle, the dimension of harmonic chains decreases by 1. 
Therefore, one can find harmonic chains that survive the insertion of triangles.

\subsection{Comparison: Persistence Barcode vs. Canonical Harmonic Chain Barcode}  
In the filtration $F_1$, we see that the canonical barcode of harmonic chains contains short bars whereas the ordinary persistence barcode has long bars. 
Our suggestion of taking canonical {\hcb} does not lead to salient cycle features for this example. 
This reflects the fact that the status of cycles, as being harmonic or not, in the first filtration $F_1$ does not provide us with a feature that spans a long time interval.

%\Bei{Meaning? This does not rule out, of course, the existence of other types of a stable cycle barcode that in $F_1$ also provides us with a stable barcode with long bars.} 

We observe that the canonical {\hcb} for $F_2$ reflects the existence of geometric features that survive the insertion of many triangles, whereas, the ordinary persistence barcode captures only two long features, if the shorter bars in the persistence barcode of~\cref{fig:example-2b} are counted as noise. 
Therefore, we see that the canonical {\hcb} can differentiate these two situations, whereas the ordinary persistence  barcode cannot (again, if we consider the shorter bars in the ordinary persistence barcode as noise).   

\subsection{Computing Canonical Barcode of Harmonic Chains}
\label{ssec:computation}

To compute the canonical barcode of harmonic chains, we make use of the following algorithm.
Let $F$ be a filtration and $\har(K_i)$ be the harmonic chain group of the complex $K_{t_i}$.

For $i = 0, 1,2, \ldots,m$, let $\chd_{i-1}$ be the part of the barcode computed up to time step $i$, with $\chd(-1) = \emptyset$. For each $j = i+1, i+2, \ldots$ compute $r_{ij} := dim (f^{i,j}(\har(K_{t_i})) \cap \har(K_{t_j}))$, where here we write again $f$ for the induced map on cycles. If there are already $m_i$ bars of $\chd_{i-1}$ alive at time $i$, and $h_i = dim(\har(K_{t_i})$, we add $\max\{h_i - m_i,0\}$ new bars to $\chd_{i-1}$ starting at $i$. These bars will die when $r_{ij}$ decreases as $j$ increases. If at an index $j$, $r_{ij} - r_{ij-1} >0$, we kill $r_{ij} - r_{ij-1}$ many harmonic bars which started at $i$. The bars we do not kill survive to $+\infty$. 
This concludes the description of the algorithm.

This computation can be done easily by first computing for each $i$, the Laplacian and a basis for the nullity of the Laplacian.  
Then for each $j>i$, we compute the matrix of images $\har^{i,j} := f^{i,j}(\har(K_{t_i}))$ simply by appending 0's to the end of a matrix of basis of $\har(K_{t_i})$ to make the number of rows equal to that of the matrix at time $j$. 
We then compute $r_{ij} = dim (f^{i,j}(\har(K_i)) \cap \har(K_j))$ by computing the rank of $ [ \har_j , \har^{i,j}]$ (i.e., the concatenation of two matrices) as an intermediate step. 

If $n$ is the number of simplices in the complex, and $m$ is the number of time steps, then all of these matrix operations can be done in time $O(m^2n^\omega + mn^3)$, where $n^\omega$ is the matrix multiplication time. In the algorithm, we compute the Laplacian and a basis for its kernel in with $n^3$ time in each time step, which is $O(mn^3)$ for all time steps. Then we need for each pair $i<j$, to compute the rank of $ [ \har_j , \har^{i,j}]$, which is possible in $m^2 n^\omega$ time. In total, we have $O(m^2n^\omega + mn^3)$. It is plausible that the runtime can be further improved. Providing the most efficient algorithm to compute the canonical {\hcb} is left for future work. We also prove that this algorithm is correct, see~\cref{sec:correctness} for details. 

\begin{theorem}
    Given a filtration of a complex of size $n$, with $m$ time-steps, we can compute the canonical barcode of harmonic chains in $O(m^2n^\omega + mn^3)$ time.
\end{theorem}

\begin{figure*}[!ht]
    \centering
    \includegraphics[scale=0.8]{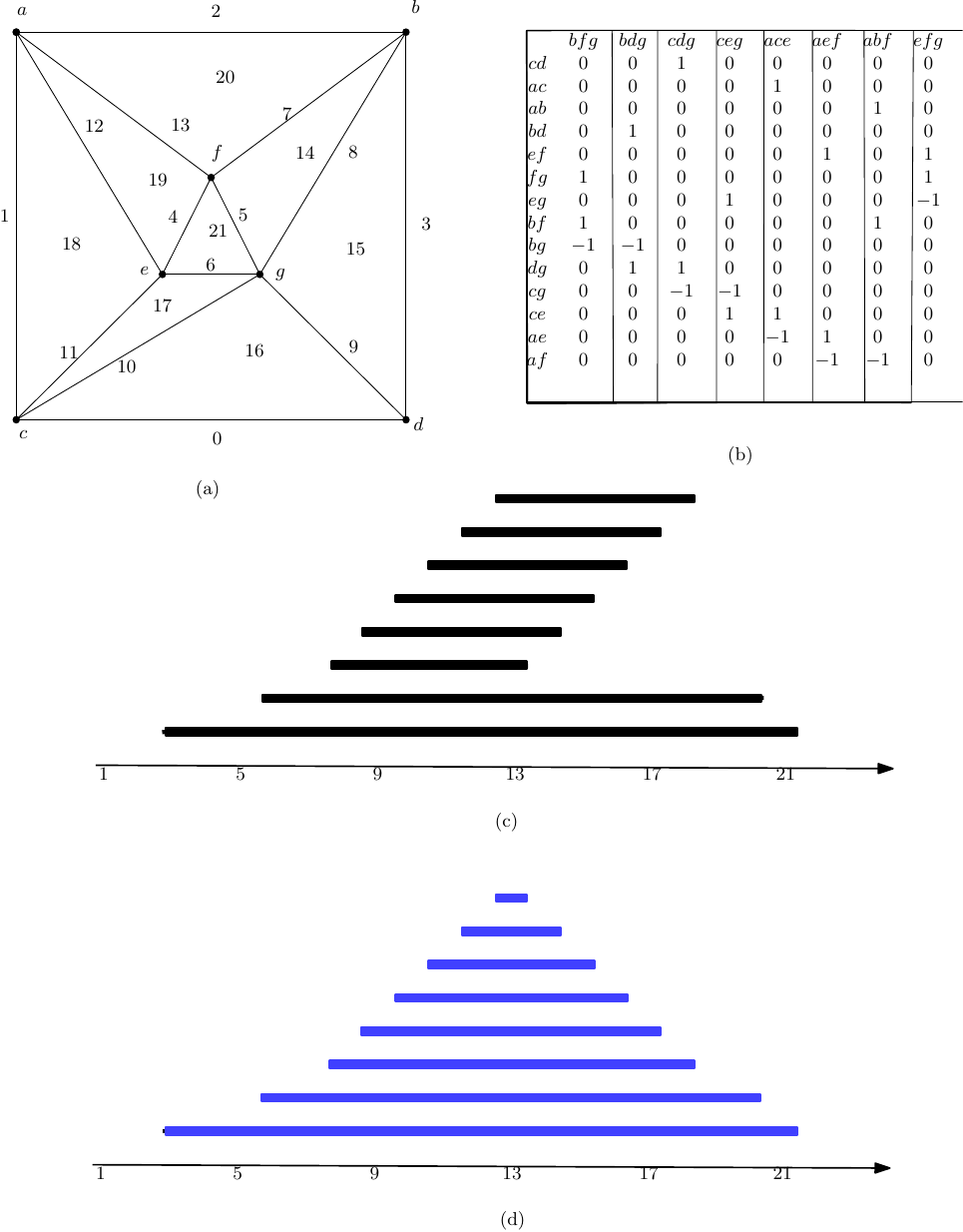}
    \caption{The ordinary persistence barcode and the canonical barcode of harmonic chains. (a) The insertion times of edges and triangles for the filtration $F_2$. (b) The 2-dimensional boundary matrix, the matrices of different times steps are marked. (c) The ordinary persistence barcode (in black). (d) The canonical barcode of harmonic chains (in blue). The persistence reduction of the boundary matrix remains the same since the matrix is unchanged. The basis for ordinary persistence barcode therefore remains the same chains as our first ordering, however, the insertion times of simplices are changed, resulting in different bars.}
    \label{fig:example-2b}
\end{figure*}

\subsection{Stability}
We express the stability of the canonical barcode by considering the typical setting in which the two filtrations $F$ and $G$ are sub-level set filtrations for two simplex-wise linear functions $\hat{f}$ and $\hat{g}$, respectively, defined on the same complex $K$. This is the same setting as the seminal stability result of~\cite{cohen2005stability}. However, we prove, as an intermediate step, a more general result. The proofs are inspired by the overall approach of \cite{cohen2005stability}. Similar to~\cite{cohen2005stability}, our stability result can be extended to functions on topological spaces, however, we limit our scope of discussion to simplicial complexes.

We start by defining the notion of harmonic interleaving distance. Then we show that the bottleneck distance between canonical {\hcbs} is upper bounded by the harmonic interleaving distance between the two filtrations; see~\cref{fig:h-interleaving}. Eventually we show that bottleneck distance is upper bounded by $||\hat{f} - \hat{g} ||_\infty$.

\begin{definition}
    Let $\{\phi_\alpha\}$,$\{ \psi_\alpha\}$, $\phi_\alpha: F_\alpha \xrightarrow{} G_{\alpha + \eps}$, $\psi_\alpha: G_\alpha \xrightarrow{} F_{\alpha + \eps}$ be a family of chain maps defining an $\eps$-chain interleaving between $F$ and $G$. We say that the pair of families $\{\phi_\alpha\}$, $\{\psi_\alpha\}$ are \textit{harmonic-preserving} if
    \begin{itemize}[noitemsep,leftmargin=*]
        \item For any harmonic chain $c\in F$, for all $\alpha$ and $\rho \geq \eps$, if $\delta(f^{\alpha, \alpha+ \rho}(c))=0$ then $\delta \phi_\alpha (c)=0$.
        \item For any harmonic chain $c'\in G$, for all $\alpha$ and $\rho \geq \eps$, if $\delta(g^{\alpha,\alpha+\rho}(c'))=0$ then $\delta \psi_\alpha (c')=0$.
    \end{itemize} 

    The harmonic interleaving distance $d_{HI}$ is defined analogously to the chain interleaving distance, however, we allow only harmonic-preserving maps.
\end{definition}

\begin{figure}
    \centering
    \includegraphics{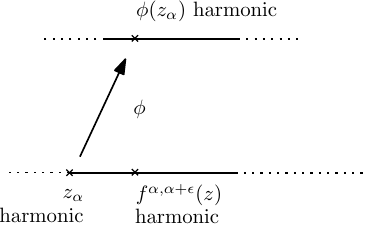}
    \caption{Harmonic interleaving at the chain level. If $z_\alpha$ is still harmonic at time $\alpha+\rho$, then $\phi(z)$ is harmonic.}
    \label{fig:h-interleaving}
\end{figure}

The following lemma is analogous to the Box Lemma in~\cite{cohen2005stability}.

\begin{lemma}
\label{lemma:box-lemma}
Let $F$ and $G$ be filtrations of complexes $K$ and $K'$, respectively, defined using inclusion maps, and let $\eps > 0$ be such that there are $\{ \phi_t \}$ and $\{ \psi_t \}$ realizing a harmonic-preserving $\eps$-interleaving (at the chain level). Let $z \in Z(K)$ be a cycle with $span(z) = [\alpha, \beta)$ such that $\alpha - \beta > \eps$. Then, writing $[\alpha', \beta')$ for $span((\phi(z)))$ we have $\alpha-\eps \leq \alpha' \leq \alpha+\eps$, and $ \beta -\eps \leq \beta' \leq \beta+\eps.$ 
\end{lemma}

\begin{proof}

In the following argument by a chain $c_t$ we mean the chain $c$ of $K$ considered as a chain of $K_t$. \cref{fig:boxlemma} helps the argument.
Consider the chain $z'=\phi_\alpha(z)$. $z'$ is clearly a cycle and since the maps are harmonic preserving, and $f^{\alpha,\alpha+2\eps}(z) \in Z_{\alpha+2\eps}(K)$ is still a cocycle, by the harmonic preserving property $z'$ is a cocycle and hence harmonic. We now consider the span of $z'$ in $G$, let it be denoted by $[\alpha', \beta')$. 

We have seen already that $\alpha'-\alpha \leq \eps$. Assume, for the sake of contradiction, that $\alpha - \alpha '>\eps$. Then $z'$ exists at time $\alpha-\eps-\eta$ for small $\eta$. We consider $y = \psi_{\alpha-\eps-\eta} (z'_{\alpha-\eps-\eta}) \in F_{\alpha-\eta}$. We suppress the $x$ in $f^{x,y}$ in the coming calculations, it can be deduced from the argument. We also suppress the subscripts of $\phi_t, \psi_t$. 

We have 
\[  f^{2\eps+\eta} (y) = f^\eta \psi \phi (y) = \psi g^\eta \phi(y) = \psi g^\eta \phi \psi (z'_{\alpha-\eps-\eta}) = \psi g^{2\eps+\eta} (z'_{\alpha-\eps-\eta}) = \psi(z'_{\alpha+\eps}) = z_{\alpha+2\eps}. \] 
Since $f$ is the inclusion, it follows that $y = z_{\alpha-\eta}$ which is a contradiction since $z$ does not exist at $\alpha-\eta$. It follows that $ \alpha-\alpha' \leq \eps$.

If $\beta'-\beta>\eps$, then $|span(z')|>\eps$ by the argument on starting points of $span(z')$ above. We can deduce that $\psi(z'_{\beta'-\eps}) = z_{\beta'}$ is harmonic by the harmonic preserving property of $\psi$. This is a contradiction since the interval of $z = \psi(z') = \psi\phi(z)$ ends at $\beta < \beta'-\eps$. Therefore, $\beta'-\beta \leq \eps$.

Assume $\beta-\beta'>\eps$. Then for small $\eta>0$, $\beta-\eps-\eta > \beta'$. It follows that $\phi(z_{\beta-\eps-\eta}) = z'_{\beta-\eta}$ is harmonic which is a contradiction, since $z'_{\beta'}$ is not harmonic and the maps are inclusions.
\end{proof}

\begin{figure}
    \centering
    \includegraphics[scale=0.8]{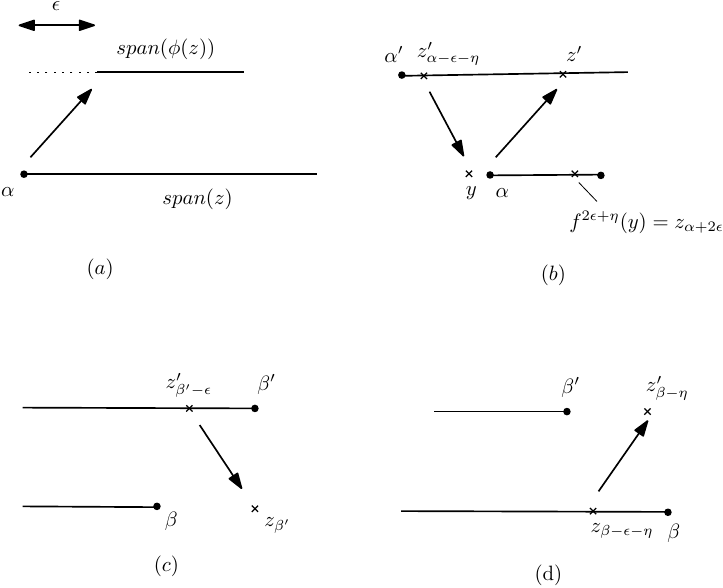}
    \caption{A schematic for the proof of~\Cref{lemma:box-lemma}.}
    \label{fig:boxlemma}
\end{figure}

For a filtration $F$, let $\delta(F) = min\{ t_{i+1} - t_{i}, i = 1, \ldots, m-1 \}$. We say that $F$ is \textit{very close} to $G$ if $d_{HI}(F,G) < \min \{ \delta(F)/4, \delta(G)/4 \} $.
Observe that $\delta(F)$ is the minimum possible distance between any two endpoints of intervals or interval lengths.

\begin{lemma}
\label{lemma:simple-stability}
If $F$ is very close to $G$ then
$$ d_B(\chd (F), \chd (G)) \leq d_{HI}(F,G).$$
\end{lemma}

\begin{proof}
Let $I_1, I_2, \ldots, I_{p_F}$ be the sequence of the intervals of $F$ in the canonical {\hcb}, and $J_1, J_2, \ldots, J_{p_G}$ be those of $G$. Moreover, let $d_1, d_2, \ldots, d_{q_F}$ be distinct lengths of the intervals $I_i$, and $l_1,l_2, \ldots, l_{q_F}$ be the lengths for $G$, in order. Let $\{ \phi_t \}$, $\{ \psi_t \}$ define a harmonic-preserving $\eps$-interleaving between $F$ and $G$, where $\eps$ is arbitrarily close to $d_{HI}(F,G)$.

Let the subspace $\Lambda_j \subset Z(K)$ be generated by cycles $z_i$ associated to $I_i$ with $|I_i| \leq  d_j$. Similarly, let $\Lambda'_j \subset Z(K')$ be the subspace generated by cycles $y_i$ associated to $J_i$ with $|J_i| \leq l_j$. We claim that for all $j$, if $d_j > 2\eps$ or $l_j > 2\eps$ then $\phi(\Lambda_j) = \Lambda'_j$  and we can match intervals of length at least $d_j$ to those of at least $l_j$.  

In the following for any cycle $z \in Z(K)$, by $\phi(z)$ we mean $\phi_\alpha(z)$ where $\alpha$ is the time that $z$ is born. We extend this convention to subspaces of cycles in the natural way. 
We first prove the claim for $j=1$. Without loss of generality, assume $ d_1 \geq l_1 $. Consider a generator $z$ of $\Lambda_1$ associated to an interval with length $d_1$. If $\phi(z) \notin \Lambda'_1$, then, \Cref{lemma:box-lemma} implies that we can add one more independent span length of length at least $ d_1 - 2\eps  \geq l_1- 2\eps \geq l_1 -\delta(G)/2 > l_2$ to the lex-maximal independent span length sequence of $G$. This is a contradiction. It follows that $\phi(\Lambda_1) \subset \Lambda'_1$. To show that $\Lambda'_1 \subset \phi (\Lambda_1)$, it is enough to argue that the generators of $\Lambda'_1$ are in the image. Since $\phi(\Lambda_1) \subset \Lambda'_1$ we have $d_1 \geq l_1 - 2\eps$. If one such generator $y$ is not in the image $\phi(\Lambda_1)$, then $|span(\psi(y))| \geq l_1 - 4\eps > d_2$, hence $\psi(y) \in \Lambda_1$. Thus $y = \phi \psi (y) \in \Lambda'_1$. This contradicts our assumption on $y$. It follows that $\lambda'_1 \subset \phi(\Lambda_1)$ and thus $\Lambda'_1 = \phi(\Lambda_1)$. 

We now show that each interval $I_i$ in the barcode of $F$ with $|I_i| = d_1$, has to be mapped by $\phi$ to an interval with length $l_1$ in the barcode of $G$. Let an interval $I$ be generated by $z$. Since $\phi(z) \in \Lambda'_1$, $\phi(z)$ can be written as a combination of generators $y_1, \ldots, y_{n_1}$ of intervals of the barcode of $G$. $\phi(z)$ has to be born at the time when one of the $y_i$ is born. It follows that $span(\phi(z))$ agrees with interval of $y_i$, since the start and end of the next nearest interval differs by at least $\delta(G)> \eps$ from start and end of interval of $\phi(z)$.

Assume we have matched the intervals of length less than $d_k$ to those of length less than $l_k$ and such that the endpoints of matched intervals are within $\eps$. We consider the intervals of length $d_k$ and $l_k$. Without loss of generality assume $d_k > l_k$. Let $z$ be a generator of length $d_k$. Since $\phi$ is injective $\phi(z)$ is independent of $\Lambda'_{k-1}$. If $d_k > l_k + 2\eps$ then the lex-maximal sequence of span lengths of $G$ can be improved by adding $\phi(z)$. Therefore, $\phi(\Lambda_k) \subset \Lambda'_k$ and $d_k \leq l_k+2\eps$ and $l_k \geq d_k -2\eps$. 
If there is a $y \notin \phi(\Lambda'_k)$ which is a generator of an interval $J_i$ of length $l_k$, then $span(\psi(y))$ is of length at least $l_k -2\eps \geq d_k-4\eps = d_k -\delta(F)$. The next span length is at most $d_k-\delta$. Therefore, $|span(\psi(z))|=d_k$ and $\psi(y)$ needs to be in $\Lambda_k$, and we reach a contradiction as before.  It follows that $\phi(\Lambda_k) = \Lambda'_k$.   

We now show that each interval $I_i$ in the barcode of $F$ with $|I_i| = d_k$, has to be mapped by $\phi$ to an interval with length $l_k$ in the barcode of $G$. Let intervals $I_1, \ldots, I_m$ of length $d_k$ be generated by $z_1, \ldots, z_m$. For each $i$, since $\phi(z_i) \in \Lambda'_k$, $\phi(z_i)$ can be written as a combination of generators $y_1, \ldots, y_{n_k}$ of intervals of the barcode of $G$ of length up to $d_k$. $\phi(z_i)$ has to be born at the starting point of some $y_{j_i}$ in the combination, say at time $t_i$. By induction hypothesis, all intervals of length longer than $l_k$ born at $t_i$ are already matched, so at time $t$ an interval of length at most $l_k$ is generated for each $i=1,\ldots,m$, otherwise, some of $\phi(z_i)$ would be in $\Lambda'_{k-1}$. Since $l_{k+1}-l_k> \delta$, it follows that an interval of length $l_k$ is generated at time $t_i$, and we can match $span(z_i)$ with some interval of length $l_k$ in the barcode of $G$ with a cost of $\eps$.

The above argument can be repeated as long as both $d_k$ and $l_k$ are present, and $\min{d_k,l_k} > \eps$ so that we can use the harmonic-preserving property of the interleaving maps. By~\Cref{lemma:box-lemma}, if $\max\{d_k, l_k\}> 2\eps$, then $\min\{l_k,d_k\}>\eps$. So the argument can be repeated as long as $\max\{d_k, l_k\} > 2\eps$. If $\max\{d_k, l_k\} \leq 2\eps$ we match the interval to the nearest point of the diagonal. This is possible again with a cost of $\eps$.
\end{proof}
The following is analogous to the Interpolation Lemma from~\cite{cohen2005stability}.

\begin{theorem}
Let $\hat{f}$ and $\hat{g}$ be two real-valued simplex-wise linear functions defined on a simplicial complex $K$. Let $F$ and $G$ be the sub-level-set filtrations of $\hat{f}$ and $\hat{g}$. Then $$ d_B (\chd(F),\chd(G)) \leq || \hat{f} -\hat{g}||_\infty.$$
\end{theorem}

\begin{proof}
The proof is based on the argument of~\cite{cohen2005stability} with some modifications. Let $c = ||\hat{f} -\hat{g}||_\infty$. We define a 1-parameter family of convex combinations $h_\lambda = (1-\lambda)\hat{f} + \lambda\hat{g}$, $\lambda \in [0,1]$. These functions interpolate $\hat{f}$ and $\hat{g}$. Since $K$ is finite, the images $\hat{f}(K)$ and $\hat{g}(K)$ are bounded. Without loss of generality, we assume that values are at least 1 and 
$\hat{f}(x), \hat{g}(x) < M$ for all $x$. 

Given a convex combination of the form $h_\lambda = (1-\lambda)\hat{f} + \lambda\hat{g}$, $\delta(\lambda):=\delta(h_\lambda)$ is positive, and we have $h_\lambda < M$. The set $C$ of open intervals $J_\lambda = \left(\frac{\lambda - \delta(\lambda)}{(1+M)4c}, \lambda + \frac{\delta(\lambda)}{(1+M)4c}\right)$ forms an open cover of $[0,1]$. Since $[0,1]$ is compact, $C$ has a finite sub-cover defined using $\lambda_1 < \lambda_2 < \ldots < \lambda_n$ such that each two consecutive intervals $J_{\lambda_i}$ and $J_{\lambda_{i+1}}$ intersect. Then,
\begin{center}
\begin{align}
\begin{split}
     \lambda_{i+1} - \lambda_i & \leq \frac{(\delta(\lambda_i)+\delta(\lambda_{i+1}))}{(1+M)4c}  \\
     & \leq \frac{\min\{\delta(\lambda_i), \delta(\lambda_{i+1})\}(1+M)}{(1+M)4c} \\
     & \leq \frac{\min\{\delta(\lambda_i),\delta( \lambda_{i+1})\}}{4c}.
\end{split}     
\end{align}
\end{center}

Since $c = ||\hat{f} -\hat{g}||_\infty$, $||h_{\lambda_i} - h_{\lambda_{i+1}} ||_\infty = c(\lambda_{i+1}-\lambda_i)=: c_\lambda.$ It follows that $||h_{\lambda_i} - h_{\lambda_{i+1}} ||_\infty \leq \min\{\delta(\lambda_i),\delta(\lambda_{i+1})\}/4.$

Let $F(\lambda_i)$ denote the sublevel set filtration of $h_{\lambda_i}.$
Since the filtrations are sublevel set filtrations of the same complex, we have that $F(\lambda_i)_{\alpha} \subset F(\lambda_{i+1})_{\alpha+c_\lambda}$, and $F(\lambda_{i+1})_\alpha \subset F(\lambda_i)_{\alpha+c_\lambda}$. Since the inclusion is harmonic-preserving, we deduce 
\[d_{HI}(F(\lambda_{i+1}),F(\lambda_{i})) \leq ||h_{\lambda_i} - h_{\lambda_{i+1}} ||_\infty = c(\lambda_{i+1}-\lambda_i) \leq \frac{\min\{\delta(\lambda_i),\delta(\lambda_{i+1})\}}{4c}.\]
It follows from~\Cref{lemma:simple-stability} that $$d_B(\chd(F_{\lambda_i}), \chd (F_{\lambda_{i+1}})) \leq d_{HI}(F(\lambda_i),F(\lambda_{i+1})) \leq ||h_{\lambda_i} - h_{\lambda_{i+1}}||_\infty.$$

The concatenation of these inequalities follows the proof of Interpolation Lemma in~\cite{cohen2005stability} verbatim leading to the statement of the theorem.
\end{proof}

\subsection{Stability Example}

We consider the same example we presented for instability, see~\cref{fig:stableexample}. 
For the left filtration, at time $t_4$, the cycle $z_1 = t_4+t_2-t_1$ is created and represents the top bar. 
At time $t_5$ the dimension of the harmonic chains increases, and one can see that the cycle $z_2 = t_2-t_1+t_3-t_5$ is harmonic at $t_5$ and survives up to $t_9$.
Among the cycles that exist at $t_5$, a 1-dimensional subspace survives $t_8$, this space is generated by $z_2$.
%\Bei{?}
%For cycles that live at time $t_5$, after adding $t_8$ 
%Among edges inserted until $t_5$, no cycle survives adding $t_8$ other than multiples of $z_2$. 
%At $t_8$, the harmonic cycle generate 
%\Bei{Please add a sentence here.}
At $t_7$, all edges are present, therefore there has to be a cycle that survives until $t_{10}$. 
We see that $z_3 = \partial(t_{10}) - 1/2 \partial (t_8 - t_9 )$ has this property and serves as the harmonic representative.

For the right filtration, at time $t_4$, the cycle $\bar{z}_1 = -t_1+t_2+t_4-t_3$ is created and survives up to $t_9$. 
At time $t_5$, $\bar{z}_2 = t_5+t_4-t_3$ is created and  survives up to $t_8$. 
Finally, at time $t_7$, $\bar{z}_3=z_3$ is created and survives up to $t_{10}$.

We observe that the first bar on top on the left can be matched to the second bar on the right, and the second bar in the left can be 
matched to the first bar on the right. The bottom bar is unchanged in both barcodes.

\begin{figure}
    \centering
    \includegraphics{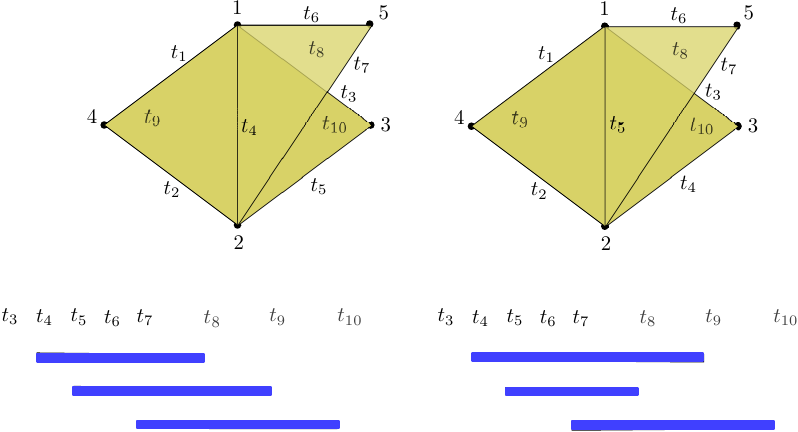}
    \caption{The canonical barcodes of harmonic chains (in blue) are stable.}
    \label{fig:stableexample}
\end{figure}

%% file: sec-conclusion.tex
\section{Future Work}
\label{sec:discussion}

We introduce a canonical barcode of harmonic chains as a novel barcode from a persistence filtration that captures geometric and topology information of data. 
In the future, such a barcode can be utilized in place of or alongside the persistence barcode. 
We will perform experiments comparing our new {\hcb} with the persistence barcode in applications, such as topology-based feature vectorization and classification.   
We also hope to perform in-depth investigation of its interpretability.

%% file: sec-correctness.tex
\section{Algorithmic Correctness for Computing Canonical Harmonic Chain Barcode}
\label{sec:correctness}

In this section, we prove that the algorithm for computing the canonical {\hcb} in~\cref{ssec:computation} is correct. 
We assume that in a given filtration $F$, we add one simplex at a time.
Under this assumption, we compute the canonical {\hcb} by sweeping time from $-\infty$ to $\infty$. 
At any time $t$, we maintain a set of cycles $R(t)$ satisfying the invariant that:
\begin{itemize}[noitemsep,leftmargin=*]
    \item The set $R(t)$ contains cycles realizing a lexicographical maximal (lex-maximal) sequence of independent span length for $Z(K_t) \subset Z(K)$.
\end{itemize}

The following Lemma describes the change from $R(t_i)$ to $R(t_{i+1})$.

\begin{lemma}
\label{lemma:greedy}
    Assume that the harmonic representatives $R(t_{i})$ satisfy the invariant at time $t_i$. Then, at time $t_{i+1}$, if the dimension of the harmonic chains increases, to obtain $R(t_{i+1})$ satisfying the invariant, it is enough to add to $R(t_{i})$ the lex-maximal set of independent spans born at $t_i$.    
\end{lemma}
\begin{proof}
Let $\sigma$ be the simplex inserted at time $t_{i+1}$. If $\sigma$ is not $d$- or $(d+1)$-dimensional then the set of $d$-cycles and harmonic $d$-cycles do not change. We then can set $R(t_{i+1})=R(t_{i})$.

Suppose $\sigma$ is $d$-dimensional. Then it might create new harmonic cycles. If this does not happen, we can set $R(t_{i+1})=R(t_{i})$. 
A new harmonic cycle is created if and only if a new $d$-homology class is born. In this case, the dimension of harmonic cycles increases. Note that $\har (K_{t_{i}}) \subset \har (K_{t_{i+1}})$ is a subspace of codimension 1.
    
Let $N$ be the maximum death time of all harmonic cycles generated at time $t_{i+1}$. $N-t_{i+1}$ is the maximum span length of these cycles and their span is determined by their death time. 
Let $\tau \in \har (K_{t_{i+1}})$ be a cycle that realizes the span $[t_{i+1},N)$. Set $R' = R(t_i) \cup \{ \tau \}$. We claim that $R'$ satisfies the invariant at time $t_{i+1}$. 

To prove the claim, we consider how the maximum span length sequence of $Z(K_{t_{i+1}})$, denoted as $\rho_1 \geq \rho_2 \geq \cdots$, is constructed. 
Recall that we first take a cycle with the maximum span length. If this cycle is in $Z(K_{t_{i}})$, then this span length will appear as span interval with maximum length in $R(t_i) \subset R'$ by the induction hypothesis. This will be again true for the second number, etc. Let $\rho_j$ be the first appearance of a number that is not in the maximum span length sequence for $Z(K_{t_{i}})$ (note that it can be that $\rho_j = \rho_{j-1}$, in this case, $\rho_j$ is one more repetition of the same number). A realizing cycle is in $Z(K_{t_{i+1}}) - Z(K_{t_{i}})$ and it has to have the maximum span length in this set. Therefore, $\rho_j = N-t_{i+1}$ and $\tau$ realizes this span length. Now we claim that the rest of the cycles in a sequence realizing the maximum span length sequence for $Z(K_{t_{i+1}})$ must come from $Z(K_{t_{i}})$ and hence their span lengths must appear among spans of $R(t_i) \subset R'$.

To prove the claim, let us assume for the sake of contradiction that a second cycle appearing in order of span lengths $\tau' \in Z(K_{t_{i+1}}) - Z(K_{t_{i}})$ has a span length in the sequence, and $|span(\tau')|\leq |span(\tau)|$. 
$\tau'$ must have a non-zero coefficient for $\sigma$, say $\beta$. Let $\alpha \neq 0$ be the coefficient of $\sigma$ in $\tau$. We can form the cycle $\chi = \alpha/ \beta \tau' - \tau \in Z(K_{t_{i}})$. When $\chi$ becomes non-harmonic, i.e., at the first $t$ such that $\delta_t(\chi)\neq 0$, at least one of $\tau'$ and $\tau$ must also become non-harmonic. Since the death time of $\tau$ is equal or larger than $\tau'$, the death time of $\chi$ is at least that of $\tau'$. Since $\chi$ already exists at time $t_{i}$ it must be that $|span(\chi)| > |span(\tau')|$. This contradicts the selection of $\tau'$ as it is in a space generated by cycles of larger span length $\tau$ and $\chi$. This finishes the argument.

If $\sigma$ is $(d+1)$-dimensional, the set of $d$-cycles $Z(K_t)$ does not change and we can set $R(t_{i+1}) = R({t_i})$.
\end{proof}

Now we consider an arbitrary filtration $F$. Let $\Sigma_i$ be the set of simplices inserted at time $i$. We can simulate $F$ using a filtration which inserts $\Sigma$ one simplex at a time $t_{i,j}=t_{i,j'}$. In other words, we apply~\Cref{lemma:greedy} one simplex at a time, and the difference betrween times is 0. We need, however, to consider all the possible orderings for insertions. For any fixed ordering,~\Cref{lemma:greedy}  applies. We need to take the ordering which leads eventually to the lex-maximal sequence of spans added at time $t_{i+1}$. We can obtain this sequence by computing the new harmonic chain space $A = \har{(K_{t_{i+1})}} - \har{(K_{t_i})}$. All these chains are born at $t_{i+1}$. We need to choose first the longest such span, then the next independent largest span, and so on. Since the span lengths in $A$ are determined by death times, this is what the algorithm of~\cref{ssec:computation} computes.

%% file: Arxiv-main.bbl
\begin{thebibliography}{10}

\bibitem{basu2022harmonic}
S.~Basu and N.~Cox.
\newblock Harmonic persistent homology.
\newblock In {\em 2021 IEEE 62nd Annual Symposium on Foundations of Computer
  Science (FOCS)}, pages 1112--1123. IEEE, 2022.

\bibitem{BubenikHullPatel2020}
P.~Bubenik, M.~Hull, D.~Patel, and B.~Whittle.
\newblock Persistent homology detects curvature.
\newblock {\em Inverse Problems}, 36(2), 2020.

\bibitem{CarlssonZomorodianCollins2004}
G.~Carlsson, A.~J. Zomorodian, A.~Collins, and L.~J. Guibas.
\newblock Persistence barcodes for shapes.
\newblock {\em Proceedings Eurographs/{ACM} {SIGGRAPH} Symposium on Geometry
  Processing}, pages 124--135, 2004.

\bibitem{ChacholskiGiuntiJin2023}
W.~Chach\'{o}lski, B.~Giunti, A.~Jin, and C.~Landi.
\newblock Decomposing filtered chain complexes: geometry behind barcoding
  algorithms.
\newblock {\em Computational Geometry: Theory and Applications}, 109, 2023.

\bibitem{chambers2022complexity}
E.~W. Chambers, S.~Parsa, and H.~Schreiber.
\newblock On complexity of computing bottleneck and lexicographic optimal
  cycles in a homology class.
\newblock In {\em 38th International Symposium on Computational Geometry (SoCG
  2022)}. Schloss Dagstuhl-Leibniz-Zentrum f{\"u}r Informatik, 2022.

\bibitem{Chazal}
F.~Chazal, D.~Cohen-Steiner, M.~Glisse, L.~J. Guibas, and S.~Y. Oudot.
\newblock Proximity of persistence modules and their diagrams.
\newblock In {\em Proceedings of the Twenty-Fifth Annual Symposium on
  Computational Geometry}, pages 237--246, New York, NY, USA, 2009. Association
  for Computing Machinery.

\bibitem{ChazalCohenSteinerGlisse2009}
F.~Chazal, D.~Cohen-Steiner, M.~Glisse, L.~J. Guibas, and S.~Y. Oudot.
\newblock Proximity of persistence modules and their diagrams.
\newblock In {\em 25th Annual Symposium on Computational Geometry (SoCG 2009)},
  pages 237--246, New York, NY, USA, 2009. Association for Computing Machinery.

\bibitem{chen2011hardness}
C.~Chen and D.~Freedman.
\newblock Hardness results for homology localization.
\newblock {\em Discrete \& Computational Geometry}, 45(3):425--448, 2011.

\bibitem{chen2022}
J.~Chen, Y.~Qiu, R.~Wang, and G.-W. Wei.
\newblock Persistent {L}aplacian projected {O}micron {BA. 4} and {BA. 5} to
  become new dominating variants.
\newblock {\em Computers in Biology and Medicine}, 151:106262, 2022.

\bibitem{cohen2005stability}
D.~Cohen-Steiner, H.~Edelsbrunner, and J.~Harer.
\newblock Stability of persistence diagrams.
\newblock In {\em Proceedings of the 21st Annual Symposium on Computational
  Geometry}, pages 263--271, 2005.

\bibitem{DeGregorioGuerraScaramuccia2021}
A.~De~Gregorio, M.~Guerra, S.~Scaramuccia, and F.~Vaccarino.
\newblock Parallel decomposition of persistence modules through interval bases.
\newblock {\em arXiv preprint arXiv:2106.11884}, 2021.

\bibitem{deSilvaVejedomo}
V.~de~Silva and M.~Vejdemo-Johansson.
\newblock Persistent cohomology and circular coordinates.
\newblock In {\em Proceedings of the Twenty-Fifth Annual Symposium on
  Computational Geometry}, SCG '09, pages 227--236, New York, NY, USA, 2009.
  Association for Computing Machinery.

\bibitem{DeyHiraniKrishnamoorthy2010}
T.~K. Dey, A.~N. Hirani, and B.~Krishnamoorthy.
\newblock Optimal homologous cycles, total unimodularity, and linear
  programming.
\newblock In {\em Proceedings of the 42nd ACM Symposium on Theory of
  Computing}, pages 221--230, 2010.

\bibitem{dey2020computing}
T.~K. Dey, T.~Hou, and S.~Mandal.
\newblock Computing minimal persistent cycles: Polynomial and hard cases.
\newblock In {\em Proceedings of the Fourteenth Annual ACM-SIAM Symposium on
  Discrete Algorithms}, pages 2587--2606. SIAM, 2020.

\bibitem{dey2022computational}
T.~K. Dey and Y.~Wang.
\newblock {\em Computational topology for data analysis}.
\newblock Cambridge University Press, 2022.

\bibitem{eckmann1944harmonische}
B.~Eckmann.
\newblock Harmonische {F}unktionen und {R}andwertaufgaben in einem {K}omplex.
\newblock {\em Commentarii Mathematici Helvetici}, 17(1):240--255, 1944.

\bibitem{edelsbrunner2002topological}
Edelsbrunner, Letscher, and Zomorodian.
\newblock Topological persistence and simplification.
\newblock {\em Discrete \& Computational Geometry}, 28:511--533, 2002.

\bibitem{edelsbrunner2010computational}
H.~Edelsbrunner and J.~Harer.
\newblock {\em Computational Topology: An Introduction}.
\newblock Applied Mathematics. American Mathematical Society, 2010.

\bibitem{Ghrist2008}
R.~Ghrist.
\newblock Barcodes: The persistent topology of data.
\newblock {\em Bullentin of the American Mathematical Society}, 45:61--75,
  2008.

\bibitem{guglielmi2023quantifying}
N.~Guglielmi, A.~Savostianov, and F.~Tudisco.
\newblock Quantifying the structural stability of simplicial homology.
\newblock {\em arXiv preprint arXiv:2301.03627}, 2023.

\bibitem{GurnariGuzman-SaenzUtro2023}
D.~Gurnari, A.~Guzm{\'a}n-S{\'a}enz, F.~Utro, A.~Bose, S.~Basu, and L.~Parida.
\newblock Probing omics data via harmonic persistent homology.
\newblock arXiv preprint arXiv:2311.06357, 2023.

\bibitem{horak2013}
D.~Horak and J.~Jost.
\newblock Spectra of combinatorial {L}aplace operators on simplicial complexes.
\newblock {\em Advances in Mathematics}, 244:303--336, 2013.

\bibitem{keros2023spectral}
A.~Keros and K.~Subr.
\newblock Spectral coarsening with {H}odge {L}aplacians.
\newblock In {\em ACM SIGGRAPH 2023 Conference Proceedings}, pages 1--11, 2023.

\bibitem{Kirchhoff1847}
G.~Kirchhoff.
\newblock Ueber die {A}ufl{\"o}sung der {G}leichungen, auf welche man bei der
  {U}ntersuchung der linearen {V}ertheilung galvanischer {S}tr{\"o}me
  gef{\"u}hrt wird.
\newblock {\em Annalen der Physik}, 148(12):497--508, 1847.

\bibitem{lee2019}
H.~Lee, M.~K. Chung, H.~Choi, H.~Kang, S.~Ha, Y.~K. Kim, and D.~S. Lee.
\newblock Harmonic holes as the submodules of brain network and network
  dissimilarity.
\newblock In {\em Computational Topology in Image Context: 7th International
  Workshop, CTIC 2019, M{\'a}laga, Spain, January 24-25, 2019, Proceedings 7},
  pages 110--122. Springer, 2019.

\bibitem{Lieutier}
A.~Lieutier.
\newblock Persistent harmonic forms.
\newblock
  \url{https://project.inria.fr/gudhi/files/2014/10/Persistent-Harmonic-Forms.pdf}.

\bibitem{liu2023algebraic}
J.~Liu, J.~Li, and J.~Wu.
\newblock The algebraic stability for persistent {L}aplacians.
\newblock {\em arXiv preprint arXiv:2302.03902}, 2023.

\bibitem{LuoHenselman-Petrusek2023}
J.~Luo and G.~Henselman-Petrusek.
\newblock Interval decomposition for persistence modules freely generated over
  {PIDs}.
\newblock arXiv preprint arXiv:2310.07971, 2023.

\bibitem{memoli2022persistent}
F.~M{\'e}moli, Z.~Wan, and Y.~Wang.
\newblock Persistent {L}aplacians: Properties, algorithms and implications.
\newblock {\em SIAM Journal on Mathematics of Data Science}, 4(2):858--884,
  2022.

\bibitem{Merris1994}
R.~Merris.
\newblock {L}aplacian matrices of graphs: a survey.
\newblock {\em Linear Algebra and its Applications}, 197-198:143--176, 1994.

\bibitem{mike2019}
J.~Mike and J.~Perea.
\newblock Multiscale geometric data analysis via {L}aplacian eigenvector
  cascading.
\newblock In {\em 2019 18th IEEE International Conference On Machine Learning
  And Applications (ICMLA)}, pages 1091--1098. IEEE, 2019.

\bibitem{mohar1991laplacian}
B.~Mohar, Y.~Alavi, G.~Chartrand, and O.~Oellermann.
\newblock The {L}aplacian spectrum of graphs.
\newblock {\em Graph theory, combinatorics, and applications}, 2(871-898):12,
  1991.

\bibitem{Morozov2005}
D.~Morozov.
\newblock Persistence algorithm takes cubic time in the worst case.
\newblock BioGeometry News, Department of Computer Science, Duke University,
  Durham, NC, 2005.

\bibitem{Obayashi}
I.~Obayashi.
\newblock Volume-optimal cycle: Tightest representative cycle of a generator in
  persistent homology.
\newblock {\em SIAM Journal on Applied Algebra and Geometry}, 2(4):508--534,
  2018.

\bibitem{wang2020persistent}
R.~Wang, D.~D. Nguyen, and G.-W. Wei.
\newblock Persistent spectral graph.
\newblock {\em International journal for numerical methods in biomedical
  engineering}, 36(9):e3376, 2020.

\bibitem{zomorodian2005topology}
A.~J. Zomorodian.
\newblock {\em Topology for computing}, volume~16.
\newblock Cambridge university press, 2005.

\end{thebibliography}
